\documentclass[12pt,draftclsnofoot, onecolumn]{IEEEtran}

\usepackage{url,color,cite}
\usepackage{caption}
\usepackage[pdftex]{graphicx}
\usepackage[cmex10]{amsmath}
\usepackage{amssymb}
\usepackage{amsthm}
\usepackage{float}
\usepackage{epstopdf}
\usepackage{algorithm,algpseudocode}
\usepackage{algpseudocode}
\usepackage{pifont}
\usepackage{subcaption}
\usepackage{cleveref}

\newtheorem{theorem}{Theorem}
\theoremstyle{definition}

\theoremstyle{definition}
\newtheorem{Example}{Example}
\newtheorem{remark}{Remark}
\theoremstyle{definition}
\IEEEoverridecommandlockouts

\usepackage{multirow}
\newtheorem{lemma}{Lemma}

\date{}
\title{Degrees of Freedom of Interference Networks with Transmitter-Side Caches}

\IEEEoverridecommandlockouts
\begin{document}

\author{\thanks{Antonious M. Girgis is with Wireless Intelligent Networks Center (WINC), Nile University, Cairo. Ozgur Ercetin is with Faculty of Engineering and Natural Sciences, Sabanci University, Istanbul, Turkey. Mohammed Nafie is with Wireless Intelligent Networks Center (WINC), Nile University, Cairo, Egypt and the Dept. of EECE, Faculty of Engineering, Cairo University, Giza, Egypt. Tamer ElBatt is with the Depart. of Computer Science and Engineering, American University, Cairo, Egypt.} 
\IEEEauthorblockN{Antonious M. Girgis, Ozgur Ercetin, Mohammed Nafie, and Tamer ElBatt}}

\maketitle
\vspace{-15mm}

\begin{abstract}
This paper studies cache-aided interference networks with arbitrary number of transmitters and receivers, whereby each transmitter has a cache memory of finite size. Each transmitter fills its cache memory from a content library of files in the placement phase. In the subsequent delivery phase, each receiver requests one of the library files, and the transmitters are responsible for delivering the requested files from their caches to the receivers. The objective is to design schemes for the placement and delivery phases to maximize the sum degrees of freedom (sum-DoF) which expresses the capacity of the interference network at the high signal-to-noise ratio regime. Our work mainly focuses on a commonly used uncoded placement strategy. We provide an information-theoretic bound on the sum-DoF for this placement strategy. We demonstrate by an example that the derived bound is tighter than the bounds existing in the literature for small cache sizes. We propose a novel delivery scheme with a higher achievable sum-DoF than those previously given in the literature. The results reveal that the reciprocal of sum-DoF decreases linearly as the transmitter cache size increases. Therefore, increasing cache sizes at transmitters translates to increasing the sum-DoF and, hence, the capacity of the interference networks. 

\end{abstract}
\vspace{-7mm}
\begin{IEEEkeywords}
Coded caching, Interference networks, Degrees of freedom, Interference management.
\end{IEEEkeywords}

\vspace{-5mm}
\section{Introduction}
 Wireless networks are experiencing an exponential growth of user traffic load due to the rapid proliferation of wireless devices. In particular, video content contributes to more than half of the overall data traffic, and it is expected that it will grow to about $75$ percent in the next few years~\cite{misc}. Caching of popular content close to user terminals has the potential to tackle this dramatic growth. In particular, on-demand video content is usually non-real-time, i.e., it is available prior to the transmission. The main idea of caching is to exploit the under-utilized network resources during off-peak hours by prefetching the most popular content in the edge servers close to users without any knowledge of the future user demands. During the peak-hours when the network is congested, the caches at the edge servers can be exploited to partly serve the user requests without accessing the central server to improve the system performance as well as the users' experience (See~\cite{el2010proactive,tadrous2016optimal,golrezaei2012femtocaching,gregori2015joint,girgis2016proactive,liu2016energy} and references therein). It was shown in~\cite{el2010proactive} that serving a portion of requests proactively and storing the data in the user caches yield a significant reduction in the outage probability. In~\cite{tadrous2016optimal}, the authors proposed a proactive caching strategy under the objective of minimizing the transmission costs. In~\cite{golrezaei2012femtocaching}, the authors developed distributed caching networks to minimize the transmission latency with constraints on the capacity of the backhaul links. The work in~\cite{gregori2015joint,girgis2016proactive,liu2016energy} aim to design caching systems under the objective of maximizing the energy efficiency of the network.
 
 Recently, Maddah-Ali and Niesen introduced an information-theoretic approach, called \textit{coded caching}, for an error-free broadcast channel in which each user is equipped with an isolated cache memory~\cite{maddah2014fundamental}. It was shown that the network traffic can be significantly reduced by jointly designing the content placement and delivery phases which in turn generates multicast coding opportunities in the delivery phase. Coded caching has received considerable attention for several settings~\cite{yu2017exact,maddah2015decentralized,niesen2017coded,ji2014average,zhang2015coded,wang2015coded,amiri2016decentralized,pedarsani2016online,shariatpanahi2016multi,karamchandani2016hierarchical}. The work in \cite{yu2017exact} characterized in closed form the memory-rate tradeoff for the coded caching problem introduced in~\cite{maddah2014fundamental} under uncoded placement. Decentralized coded caching was introduced in~\cite{maddah2015decentralized}, where each user randomly stores some bits from each file independently of each other. In~\cite{niesen2017coded,ji2014average,zhang2015coded}, the authors studied the coded caching problem for an arbitrary file popularity. The work in~\cite{wang2015coded} and~\cite{amiri2016decentralized} studied decentralized coded caching under heterogeneous cache sizes at the users. Additionally, the concept of coded caching was extended for online caching systems in~\cite{pedarsani2016online}, multi-server networks in~\cite{shariatpanahi2016multi}, and hierarchical networks in~\cite{karamchandani2016hierarchical}.

The role of caches at the transmitters of interference networks was first studied in~\cite{maddah2015cache}. It was shown that the caches at transmitters can improve the sum degrees of freedom (sum-DoF) of the network by allowing cooperation between transmitters for interference mitigation. However, in~\cite{maddah2015cache}, the authors considered only the case of three transmitters and three receivers. A lower bound on the transmission latency of cache-aided interference networks was derived in~\cite{sengupta2016cache}, where the authors propose the normalized delivery time (NDT) as a performance metric which is proportional to the inverse of the sum-DoF. In~\cite{tandon2016cloud}, the authors derived the exact characterization of the NDT of Fog radio access networks (F-RAN) with caches equipped at each transmitter. The extension for an arbitrary number of transmitters and receivers was studied in~\cite{sengupta2016cloud}. The authors in~\cite{azimi2016fundamental} investigated the NDT of the binary-fading interference channel with two receivers served by a small-cell base station with a finite-size cache and a macro base station. The problem of interference networks with caches at both the transmitters and the receivers was studied in~\cite{girgis2017decentralized,naderializadeh2017fundamental,hachem2016degrees,xu2017fundamental}. In~\cite{girgis2017decentralized}, the NDT of an F-RAN with two transmitters and an arbitrary number of receivers was characterized under a decentralized content placement scheme, where it was shown that the proposed coding scheme achieves a performance comparable to the derived lower bound. The authors in~\cite{naderializadeh2017fundamental} focused on designing one-shot linear delivery schemes for cache-aided interference networks in which the channel extension is not allowed. In~\cite{hachem2016degrees}, the approximate characterization of the sum-DoF for interference channels with caches at both the transmitters and receivers was investigated wherein the authors proposed a separation strategy between the physical and the network layers. In~\cite{xu2017fundamental}, the authors proposed an order-optimal delivery scheme for minimizing the NDT of cache-aided interference networks.

\begin{figure}[t]
\centering{\includegraphics[scale=0.4,trim=4 4 4 4,clip]{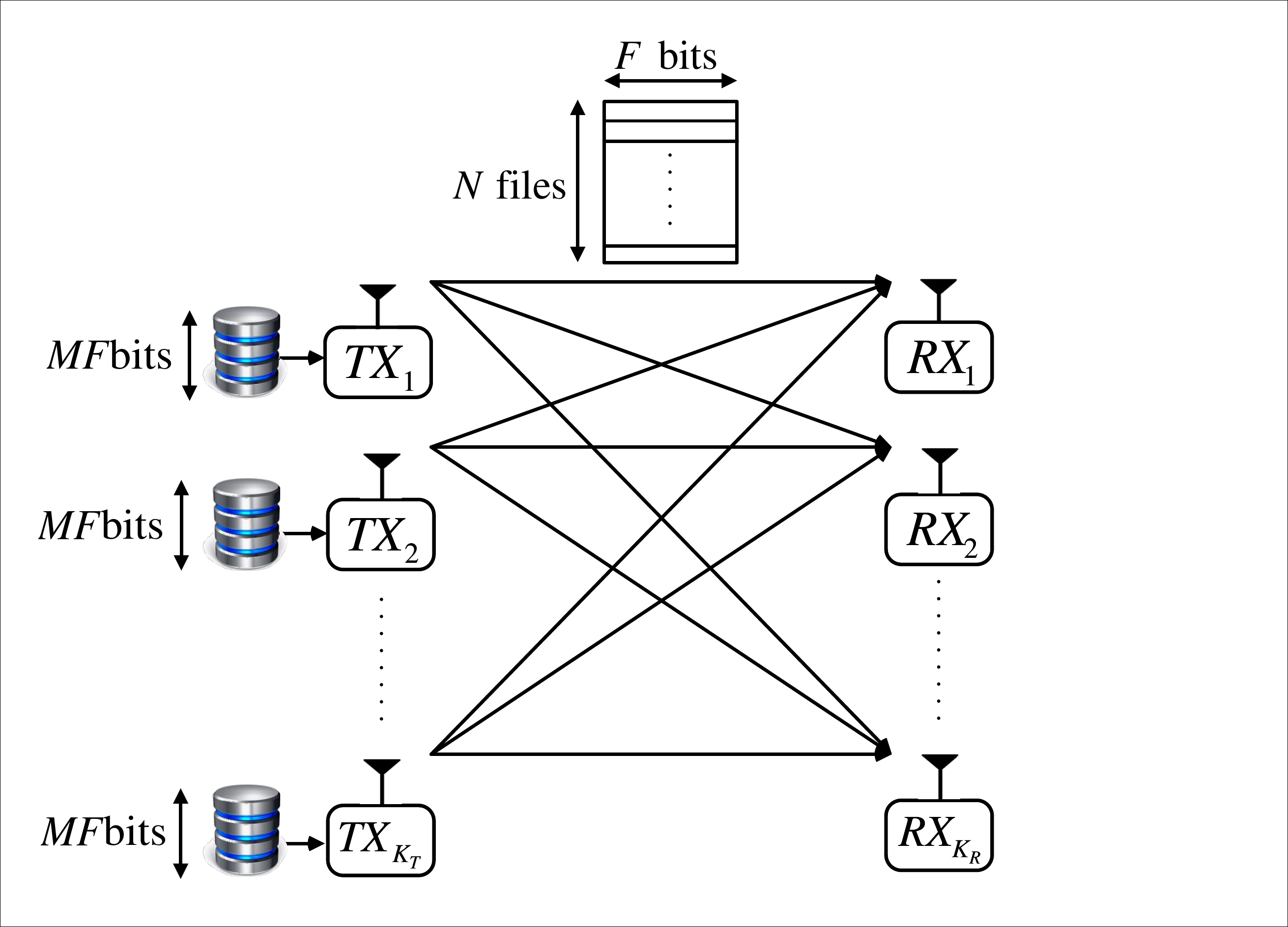}}
\caption{Cache-aided interference network with $K_T$ transmitters and $K_R$ receivers.}
\label{fig1}\vspace{-4mm}
\end{figure}

In this paper, we consider a $K_T\times K_R$ interference network with a library of $N$ files, where each transmitter is equipped with a cache memory of size $MF$ bits. We study the fundamental limits on the sum-DoF of interference networks as a function of the transmitter cache size $M$, under an uncoded placement strategy that was originally introduced in~\cite{maddah2015cache} and~\cite{naderializadeh2017fundamental}. Our main contributions in this work are as follows:

$\bullet$  We introduce a new communication paradigm in the delivery phase called \textit{cooperative $X$-networks} in which each $\frac{K_TM}{N}$ transmitters has a dedicated message for every receiver. We derive an information-theoretic bound on the sum-DoF for such network by using genie-aided, cut-set arguments. Hence, there is no delivery scheme for cache-aided interference networks achieving a sum-DoF higher than that of the derived upper bound using such uncoded caching scheme. 

$\bullet$ This derived bound gives important insights on the sum-DoF of the cache-aided interference networks. We show that the delivery scheme in~\cite{xu2017fundamental} is optimal in the case $\tau=K_R-1$, $\tau=\frac{K_TM}{N}$. Additionally, the derived bound is tighter than the upper bound in~\cite{sengupta2016cache} (developed for arbitrary caching schemes) for small cache sizes. Hence, by using an uncoded caching scheme, the bound in~\cite{sengupta2016cache} is not achievable in general. 

$\bullet$ We propose a novel delivery scheme based on zero forcing (ZF) and interference alignment (IA) techniques. The main idea of the proposed delivery scheme is to first schedule the requested bits into groups such that each group of bits (messages) would be sent in the same transmission blocks. The scheduling strategy is designed in such a way that each transmitter has a dedicated message to every receiver. In addition, the messages assigned to a transmitter is available as a side information at $\frac{K_TM}{N}-1$ consecutive transmitters, i.e., those are called cooperative (cognitive) transmitters. The benefits of cognitive messages on the sum-DoF have been studied in~\cite{lapidoth2007cognitive,annapureddy2012degrees,huang2009degrees} for interference networks, and in~\cite{jafar2008degrees} for multiple-input-multiple-output (MIMO) $X$-networks with two transmitters and two receivers. Therefore, in each transmission block, we design two-layer precoding matrices at each transmitter. The first precoding layer is designed based on ZF to leverage the partial cooperation between transmitters in order to null out some interference signals at each receiver, while, the second precoding layer is designed based on IA to align the remaining interference signals at each receiver into smaller dimensions.


$\bullet$ Although, the work in~\cite{naderializadeh2017fundamental,hachem2016degrees,xu2017fundamental} have studied a similar setup, we show that our proposed coding scheme achieves a higher sum-DoF than those of state-of-the-art schemes, where the differences between our results and those in~\cite{naderializadeh2017fundamental,hachem2016degrees,xu2017fundamental} are discussed throughout the paper. Our analytical analysis indicates that the reciprocal of the sum-DoF decreases linearly as the transmitter cache size increases. Moreover, we show that the achievable sum-DoF is within a multiplicative factor of $2$ from the upper bound which is derived without any restrictions on the caching and delivery schemes, independent of system parameters.

The rest of the paper is organized as follows. In Section~\ref{Sec1}, we introduce the system model and the problem formulation. In Section~\ref{Sec}, the main results of the paper are presented. A rigorous comparison between our coding scheme and the state-of-the-art schemes is also presented in Section~\ref{Sec}. Section~\ref{Sec2} describes the uncoded placement. The converse results for cache-aided interference networks with this uncoded placement strategy is presented in Section~\ref{Sec2}. We prove the achievable bound of the sum-DoF for a generic cache-aided interference network in Section~\ref{Sec3}. Finally, the paper is concluded in Section~\ref{Sec4}. Some technical proofs are relegated to appendices.
 
\vspace{-5mm} 
 
\section{System Model}~\label{Sec1} \vspace{-5mm}

\textbf{Notations:} The set of natural numbers and complex numbers are denoted by $\mathbb{N}^{+}$ and $\mathbb{C}$, respectively. For a square matrix $\mathbf{M}$, we use $\left(\mathbf{M}\right)^{-1}$ and $\mathit{det}\mathbf{M}$ to denote the inverse and the determinant, respectively. $\left(•\right)^{T}$ denotes transpose. We use calligraphic symbols for sets, e.g., $\mathcal{S}$. $|\mathcal{S}|$ denotes the cardinality of set $\mathcal{S}$. $\left[K\right]$ denotes the set of integers $\lbrace 1,\ldots,K\rbrace$, and $\left[a:b\right]$ denotes the set of integer numbers between $a$ and $b$: $\lbrace a,\ldots,b\rbrace$ for $b\geq a$. For set $\left[i:j\right]$ of transmitters, the indices are taken to modulo $K_T$ which is the total number of transmitters such that $\left[i:j\right]\subseteq\left[K_T\right]$. Similarly, for set $\left[i:j\right]$ of receivers, the indices are taken to modulo $K_R$ which is the total number of receivers such that $\left[i:j\right]\subseteq\left[K_R\right]$.

We consider an interference network of $K_T$ transmitters connected to $K_R$ receivers over a time-varying Gaussian channel as depicted in Figure~\ref{fig1}. There is a content library of $N$ files, $\mathcal{W}\triangleq\left\{W_1,\ldots,W_N\right\}$, each of size $F$ bits, where each file $W_f\in\mathcal{W}$ is chosen independently and uniformly from $\left[2^F\right]$ at random. Each transmitter $\text{TX}_i$, $i\in\left[K_T\right]$, has a local cache memory $Z_i$ of size $MF$ bits, where $M$ refers to the memory size in files, where $M\leq N$.

The system operates in two separate phases, a \textit{placement phase} and a \textit{delivery phase}. In the placement phase, the transmitters have access to the content library $\mathcal{W}$, and hence, each transmitter fills its cache memory as an arbitrary function of the content library $\mathcal{W}$ under its cache size constraint. We emphasize that the caching decisions are taken without any prior knowledge of the future receiver demands and channel coefficients between the transmitters (TXs) and the receivers (RXs).

 In the delivery phase, receiver $\text{RX}_j$ requests a file $W_{d_j}$ out of $N$ files of the library. We consider $\mathbf{d}=\left[d_1,\ldots,d_{K_R}\right]\in\left[N\right]^{K_R}$ as the vector of receiver demands. The transmitters are aware of all receiver demands $\mathbb{d}$. Thus, transmitter $\text{TX}_i$, $i\in\left[K_T\right]$, responds by sending a codeword $\mathbf{X}_i\triangleq\left(X_i\left(t\right)\right)_{t=1}^{T}$ of block length $T$ over the interference channel, where $X_i\left(t\right)\in\mathbb{C}$ is the transmitted signal of transmitter $\text{TX}_i$ at time $t\in\left[T\right]$. We impose an average transmit power constraint over the channel input $\frac{1}{T}||\mathbf{X}_i||\leq P$. In this phase, each transmitter has only access to its own cache contents, so that, codeword $\mathbf{X}_i$ is determined by an encoding function of the receiver demands $\mathbf{d}$, the cache contents $Z_i$, and the channel coefficients between TXs and RXs. Afterwards, each receiver $\text{RX}_j$ implements a decoding function to estimate the requested file $\hat{W}_{d_j}$ from the received signal $\mathbf{Y}_j\triangleq\left(Y_j\left(t\right)\right)_{t=1}^{T}$ which is given by
\begin{equation}~\label{eqn1}
 Y_j\left(t\right)=\sum_{i=1}^{K_T}h_{ji}\left(t\right)X_i\left(t\right)+N_j\left(t\right),
\end{equation}
where $Y_j\left(t\right)\in\mathbb{C}$ is the received signal by receiver $\text{RX}_j$ at time $t\in\left[T\right]$, and $N_j\left(t\right)$ is the additive white Gaussian noise at receiver $\text{RX}_j$ at time $t\in\left[T\right]$. Let $h_{ji}\left(t\right)\in\mathbb{C}$ represent the channel gain between transmitter $\text{TX}_i$ and receiver $\text{RX}_j$ at time $t$. We assume all channel coefficients are drawn independently and identically distributed (i.i.d.) from a continuous distribution. For a given coding scheme (caching, encoding, and decoding), the probability of error is obtained by
\begin{equation}
\mathsf{Pe}=\max_{\mathbf{d}\in\left[N\right]^{K_R}}\max_{j\in\left[K_R\right]} \mathbb{P}\left(\hat{W}_{d_j}\neq W_{d_j}\right),
\end{equation} 
which is the worst-case probability of error over all possible demands $\mathbf{d}$ and over all receivers. The transmission rate for a given coding scheme for a fixed cache size of $MF$ bits, and power constraint $P$ is $R\left(M,P\right)=K_RF/T$, where $K_RF$ is the number of bits required to be delivered to receivers and $T$ is the time required to send these bits. We say that the rate $R\left(M,P\right)$ is achievable if and only if there exists a coding scheme that achieves $R\left(M,P\right)$ such that $\mathsf{Pe}\to 0$ as $F\to\infty$. The optimum rate $R^{*}\left(M,P\right)$ is defined as the supremum of all achievable rates. Furthermore, we define the sum degrees of freedom of the network as a function of the cache size $M$ by
\begin{equation}
\mathsf{DoF}\left(M\right)=\lim_{P\to \infty}\frac{R^{*}\left(M,P\right)}{F\log\left(P\right)}.
\end{equation}

 The sum-DoF is a performance metric that defines the pre-log capacity or the multiplexing gain of the network. In other words, the capacity can be expressed by $R^{*}\left(M,P\right)=\mathsf{DoF}\log\left(P\right)+o\left(\log\left(P\right)\right)$ at the high Signal-to-Noise-Ration (SNR) regime, where the $o\left(\log\left(P\right)\right)$ term vanishes as $P\to\infty$. Our objective in this work is to characterize the trade-off between the cache size at the transmitter and the sum degrees of freedom of the cache-aided interference network. Note that the transmitter cache size should satisfy $K_TM\geq N$ to guarantee that every bit of the library content is stored, at least, at one of the transmitter caches in the network. Moreover, if the cache size is larger than the library size, i.e., $M>N$, then each transmitter is able to cache all files and the remaining cache memory would not be used, and hence, the problem will be meaningless. Therefore, we are interested in characterizing the $\mathsf{DoF}$ of the network for cache sizes such that $\frac{N}{K_T}\leq M\leq N$.
\vspace{-5mm}

\section{Main Results}~\label{Sec}\vspace{-5mm}

In this section, we present the main results of this paper. Our main contribution is proposing a novel delivery scheme for cache-aided interference networks that outperforms the coding schemes of prior work in the literature, for almost all values of the cache size. Moreover, we show that the proposed scheme is within a  multiplicative factor of $2$ from the information-theoretic bound, independent of system parameters.

\begin{theorem}~\label{Th1} For a general $K_T\times K_R$ cache-aided interference network with a library of $N\geq K_R$ files, and cache size $M\in\left[N/K_T:N\right]$ files at each transmitter, the sum degrees of freedom satisfies 
\begin{equation}
\min\left\{\frac{K_TK_R}{K_T+K_R-K_TM/N},K_R\right\}\leq\mathsf{DoF}\left(M\right)\leq \min\left\{2\frac{K_TK_R}{K_T+K_R-K_TM/N},K_R\right\}.
\end{equation}
\end{theorem}
The proposed coding scheme that achieves the lower bound on the sum-DoF in Theorem~\ref{Th1} is obtained by applying both zero-forcing and interference alignment techniques. In particular, we schedule the bits needed to be transmitted to receivers as groups, where each group contains $K_RK_T$ messages transmitted simultaneously in the same transmission block such that each transmitter has a dedicated message intended to each receiver. Furthermore, the messages assigned to a specific transmitter are available at $K_TM/N-1$ adjacent transmitters which we call cooperative or, equivalently cognitive transmitters. Thus, we design transmit beamforming matrices comprising of two layers of precoders, i.e., two precoders are multiplied by each other. The first layer (ZF precoder) is designed to null out messages intended to a specific receiver at $K_TM/N-1$ unintended, neighboring receivers. The second layer (IA precoder) is designed to align the remaining interference signals at each receiver to occupy smaller dimensions of space. The detailed description of the achievable coding scheme is presented in Section~\ref{Sec3}. Theorem~\ref{Th1} also shows that the proposed coding scheme is within at most a multiplicative factor of $2$ from the information-theoretic bound for all system parameters (the number of transmitters $K_T$, the number of receivers $K_R$, the number of files $N$, and the cache size $M$). The proof is provided in Appendix~\ref{App2}.

Before going into the details of our proposed scheme, we compare the achievable sum-DoF of our proposed coding scheme with those of the coding schemes in prior work~\cite{maddah2015cache,naderializadeh2017fundamental,hachem2016degrees,xu2017fundamental}\footnote{It is worth mentioning that \cite{naderializadeh2017fundamental,hachem2016degrees,xu2017fundamental} investigate the problem of a general interference channel with caches at both the transmitters and the receivers. Therefore, to compare their coding schemes to our proposed one, we consider the special case of their results in which the caches are available only at the transmitters, i.e., the caches at receivers are equal to zero.}. It is convenient to define a parameter $\tau=K_TM/N$ to compare between the different schemes at corner points $M=\tau N/K_T$ for $\tau\in\left[K_T\right]$. The reciprocal sum degrees of freedom $1/\mathsf{DoF}\left(M\right)$ for $M\in\left[N/K_T:N\right]$ is obtained from the lower convex envelope of these corner points, since $1/\mathsf{DoF}\left(M\right)$ is a convex function of cache size $M$ (See~\cite[Lemma~$1$]{maddah2015cache}).

In~\cite{maddah2015cache}, the authors have exactly characterized the sum-DoF of cache-aided interference networks for the special case of $K_T=K_R=3$ at corner points $\tau\in\left[3\right]$. When $\tau\in\left\{1,3\right\}$, the coding scheme in~\cite{maddah2015cache} and our proposed scheme achieve the same sum-DoF, where the optimal sum-DoF can be achieved by only using interference alignment for $\tau=1$, and the zero-forcing for $\tau=3$. When $\tau=K_R-1=2$, the authors in~\cite{maddah2015cache} have proposed a coding scheme that jointly uses both the zero-forcing and interference alignment techniques to achieve $\mathsf{DoF}=18/7$. Meanwhile our proposed scheme achieves $\mathsf{DoF}=9/4$ by using a novel scheduling technique at the delivery phase. Although the scheme in~\cite{maddah2015cache} is optimal at the point $\tau=K_R-1$, this scheme cannot be generalized for the other corner points $1<\tau< K_R-1$ for a general $K_T\times K_R$ interference network as we will elaborate more in Section~\ref{Sec2}.       

For comparison, we summarize the achievable sum-DoF of our proposed scheme and the schemes in~\cite{naderializadeh2017fundamental,hachem2016degrees,xu2017fundamental} as follows \small 
\begin{IEEEeqnarray}{lll}
\mathsf{DoF}^{\text{proposed}}&=&\min\left\{\frac{K_TK_R}{K_T+K_R-\tau},K_R\right\},~\label{eqn18}\\
\mathsf{DoF}^{\left[25\right]}&=&\min\left\{\tau,K_R\right\},~\label{eqn15}\\
\mathsf{DoF}^{\left[26\right]}&=&\frac{K_TK_R}{K_T+K_R-1},~\label{eqn16}\\
\mathsf{DoF}^{\left[27\right]}&=&\left\{\begin{array}{ll}
K_R & \tau\geq K_R\\
\frac{\tau{K_T\choose \tau}K_R}{\tau{K_T\choose \tau}+1} & \tau=K_R-1\\
\max\left\{\underset{{1\leq\grave{\tau}\leq\tau}}{\max}\ \frac{\grave{\tau}{K_T\choose \grave{\tau}}K_R}{\grave{\tau}{K_T\choose \grave{\tau}}+\left(K_R-\grave{\tau}\right){K_T\choose \grave{\tau}-1}},\tau\right\} &\tau<K_R-1
\end{array}
\right. .~\label{eqn17}
\end{IEEEeqnarray}

\normalsize
In~\cite{naderializadeh2017fundamental}, the authors have introduced a one-shot linear scheme to achieve the sum-DoF in~\eqref{eqn15} by applying the zero-forcing technique. It is easy to show that our proposed coding scheme achieves the same sum-DoF when $\tau\geq\min\left\{K_T,K_R\right\}$. Actually, the coding schemes that are based on zero-forcing only can achieve the optimal DoF as $\tau\geq\min\left\{K_T,K_R\right\}$ (See Appendix~\ref{App2}). However, our proposed coding scheme achieves a strictly higher sum-DoF when $1\leq\tau<\min\left\{K_T,K_R\right\}$ which indicates that using zero-forcing alone is not sufficient to achieve a good performance for all cache size $M$.

A delivery scheme using only interference alignment has been introduced in~\cite{hachem2016degrees} to get the sum-DoF in~\eqref{eqn16}, where the authors have neglected the gain that can be obtained from the cooperation between the transmitters when the cache size increases. In contrast, our coding scheme achieves a strictly higher DoF when $1<\tau\leq K_T$. Moreover, the multiplicative gap between the achievable sum-DoFs is given by $\frac{\mathsf{DoF}^{\text{prposed}}}{\mathsf{DoF}^{\left[26\right]}}=1+\frac{\tau-1}{K_T+K_R-\tau}$. Observe that this gap is approximately small when the number of transmitters or receivers is large. This is because, when $K_T\to\infty$ or $K_R\to\infty$, the delivery scheme based on interference alignment achieves approximately the optimal DoF of the $K_T\times K_R$ interference network which is $\min\left\{K_T,K_R\right\}$. However, at a small or moderate number of transmitters or receivers, this gap increases as the transmitter cache size increases, and hence, using both zero-forcing and interference alignment become mandatory to obtain a better performance.

Finally, \cite{xu2017fundamental} has developed a delivery scheme using both zero-forcing and interference alignment to obtain the sum-DoF in~\eqref{eqn17}. We show in Section~\ref{Sec2} that the scheme in~\cite{xu2017fundamental} achieves the optimal sum-DoF when $\tau=K_R-1$. However, our proposed scheme outperforms their delivery scheme when $1<\tau<K_R-1$. To see this, we multiply the numerator and denominator of the achievable sum-DoF in~\eqref{eqn18} with ${K_T-1\choose \tau-1}$ to get \small$\mathsf{DoF}^{\text{proposed}}=\frac{\tau{K_T\choose\tau}K_R}{\tau{K_T\choose\tau}+\left(K_R-\tau\right){K_T-1\choose\tau-1}}$. \normalsize The reason behind this is that our delivery scheme is designed to efficiently minimize the space dimensions spanned by the interference signals at each receiver, compared to the scheme in~\cite{xu2017fundamental}.

\begin{figure}[t]
\centering{\includegraphics[scale=0.37,trim=4 4 4 4,clip]{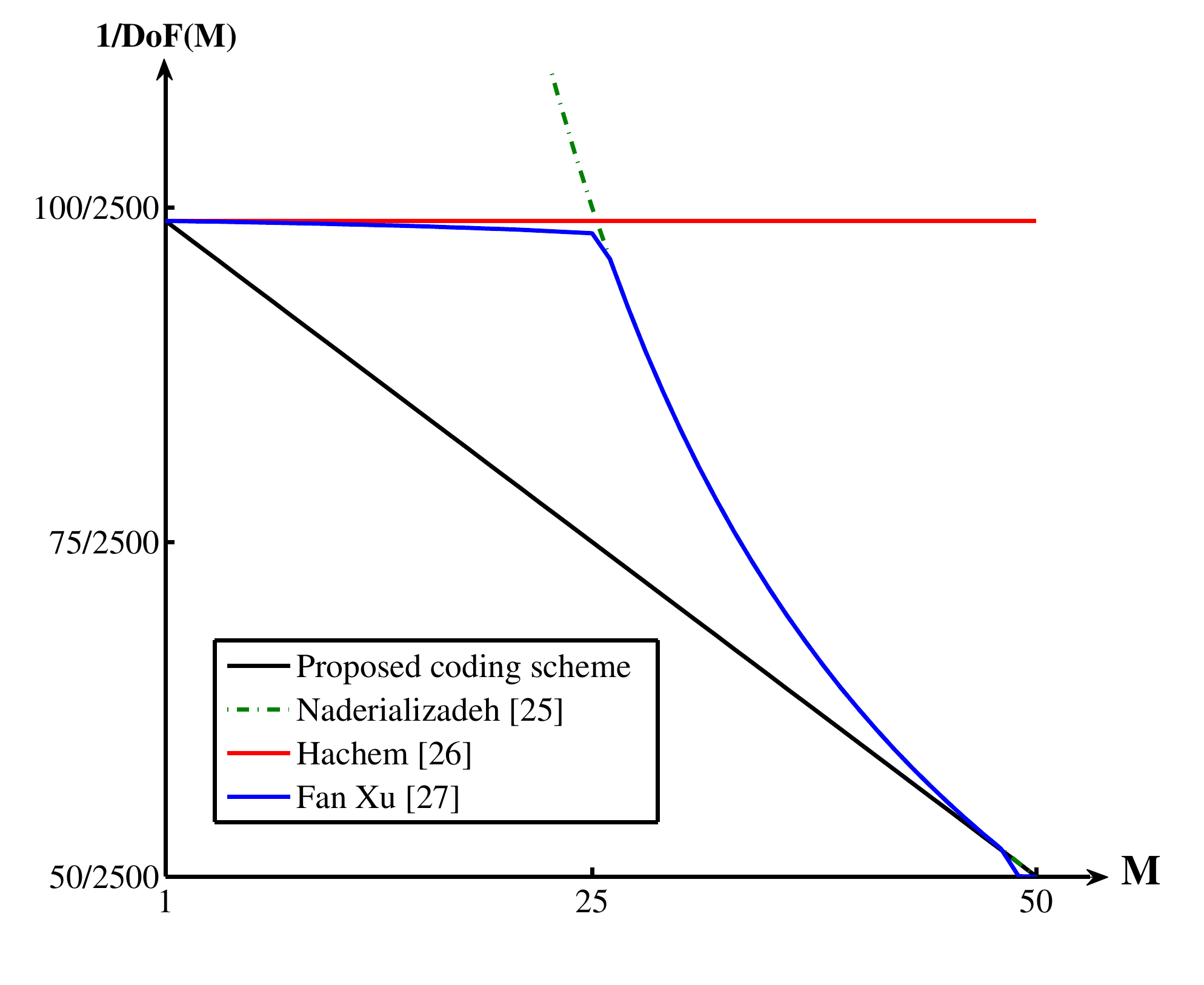}}
\caption{Comparison of the reciprocal sum degrees of freedom versus the transmitter cache size for interference network with $K_T=50$ transmitters and $K_R=50$ receivers.}
\label{fig2}\vspace{-5mm}
\end{figure}
   
\begin{Example} Consider an interference network with $K_T=50$ transmitters, $K_R=50$ receivers, and a library of $N=50$ files. Figure~\ref{fig2} shows the inverse of the sum-DoF of our proposed scheme and the other schemes in~\cite{naderializadeh2017fundamental,hachem2016degrees,xu2017fundamental} as a function of the transmitter cache size. It is observed that the proposed coding scheme outperforms the other schemes, specifically for moderate cache sizes. For example, when the transmitter cache size is $M=25$ files, our proposed scheme achieves $\mathsf{DoF}^{\text{proposed}}=33.33$, while ZF scheme in~\cite{naderializadeh2017fundamental} achieves $\mathsf{DoF}^{\left[25\right]}=25$, IA scheme in~\cite{hachem2016degrees} achieves $\mathsf{DoF}^{\left[26\right]}=25.25$, and the delivery scheme in~\cite{xu2017fundamental} achieves $\mathsf{DoF}^{\left[27\right]}=25.49$.
\end{Example}

\section{Content Placement and The Cooperative $X$-Network}\label{Sec2}

In this section, we first present the content placement strategy for cache-aided interference networks that was introduced in~\cite{maddah2015cache} and~\cite{naderializadeh2017fundamental}. Then, we show that for this content placement strategy, there arises a new network topology in the delivery phase which we call the \textit{cooperative $X$-network}. For the cooperative $X$-network, we derive an upper bound on the sum-DoF which is considered an upper bound for any delivery scheme using such uncoded placement strategy.

\subsection{Content Placement}~\label{Subsec2}
\vspace{-5mm}

Assume each transmitter has a cache of size $M=\tau N/K_T$, where $\tau\in\left[K_T\right]$. In the placement phase, we split each file $W_f\in\mathcal{W}$ into ${K_T\choose \tau}$ disjoint subfiles, each of size $F/{K_T\choose \tau}$ bits. As a result, file $W_f$ is represented by
\begin{equation}
W_f=\left\{W_{f,\mathcal{S}}:\mathcal{S}\subseteq \left[K_T\right], |\mathcal{S}|=\tau\right\}.
\end{equation}
For every file $W_f\in\mathcal{W}$, the subfile $W_{f,\mathcal{S}}$ is stored at the cache of transmitter $\text{TX}_i$, if $i\in\mathcal{S}$. Thus, each transmitter stores $N{K_T-1 \choose \tau-1}$ subfiles. Accordingly, the number of bits stored at each transmitter is equal to $N{K_T-1 \choose \tau-1}\frac{F}{{K_T\choose \tau}}=MF$ bits. Note that this placement strategy satisfies the cache size constraint for each transmitter. Moreover, we emphasize that the content placement is performed without any prior knowledge of the receiver demands or channel gains in the delivery phase, which is a practically relevant assumption, since there is a large time separation between the placement and delivery phases. This placement strategy has been first proposed in~\cite{maddah2015cache} and~\cite{naderializadeh2017fundamental} for cache-aided interference networks. For an illustration of the content placement, please refer to Example~$1$ in Subsection~\ref{Subsec1}.
 
\subsection{Cooperative $X$-network}~\label{Subsec3} 
\vspace{-5mm}
 
In the delivery phase, consider a demand vector $\mathbf{d}$, where the receiver $\text{RX}_j$ requests file $W_{d_j}$, $j\in\left[K_R\right]$. The transmitters should send the requested subfiles \small$\left\{W_{d_j,\mathcal{S}}:j\in\left[K_R\right],\mathcal{S}\subseteq\left[K_T\right],|\mathcal{S}|=\tau\right\}$ \normalsize to the receivers with a total of $K_R{K_T\choose \tau}$ subfiles. In this case, each subset $\mathcal{S}$ of $\tau$ transmitters has a dedicated message $W_{d_j,\mathcal{S}}$ to be delivered to the receiver $\text{RX}_j$, $j\in\left[K_R\right]$. This transmission problem is called a \textit{ cooperative $X$-network}, where $\tau$ is the cooperation order between transmitters. Note that the cooperative $X$-network formulation involves several types of transmission problems discussed next:

 $\bullet$ (\textit{$X$-networks}): When $\tau=1$, the problem is reduced to the traditional $X$-network studied in~\cite{cadambe2009interference} in which each transmitter has a dedicated message for every receiver with no cooperation between the transmitters. The optimal sum-DoF of $K_T\times K_R$ $X$-networks is given by $\frac{K_TK_R}{K_R+K_R-1}$ that is obtained by employing interference alignment.
 
 $\bullet$ (\textit{MISO broadcast channel}): When $\tau=K_T$, all transmitters fully cooperate with each other to deliver $K_R$ messages, one for each receiver. Thus, the problem becomes a MISO broadcast channel with a single transmitter of $K_T$ antennas connected to $K_R$ receivers, each equipped with a single antenna. The maximum sum-DoF of a MISO broadcast channel is given by $\min\left\{K_T,K_R\right\}$ which is obtained using zero-forcing~\cite{weingarten2006capacity}.
 
 $\bullet$ (\textit{MISO $X$-networks}): For general $1<\tau<K_T$, observe that the messages $\lbrace W_{d_j,\mathcal{S}}\rbrace_{j=1}^{K_R}$ are available at $\tau$ transmitters $\lbrace \text{TX}_i\rbrace$, $i\in\mathcal{S}$, i.e., these messages can be transmitted using $\tau$ antennas. Hence, the cooperative $X$-network can be seen as a MISO $X$-network with ${K_T\choose \tau}$ virtual transmitters each with $\tau$ antennas and $K_R$ single-antenna receivers. The MISO $X$-network was studied in~\cite{sun2012degrees}, where the optimal sum-DoF for ${K_T\choose \tau}$ transmitters each with $\tau$ antennas and $K_R$ single-antenna receivers is determined as \small$\frac{\tau {K_T\choose \tau} K_R}{\tau {K_T\choose \tau}+K_R-\tau}$. \normalsize However, there is a subtle, yet, important difference between the cooperative $X$-network and the actual MISO $X$-network. The channels between the multi-antenna transmitters and receivers in the actual MISO $X$-network are i.i.d. drawn from continuous distributions; however, the channels between the ${K_T\choose \tau}$ virtual transmitters and $K_R$ receivers in the cooperative $X$-network are dependent random variables derived from only $K_R K_T$ actual channels. Therefore, this raises an unsettled key question: Does the cooperative $X$-network with cooperation order $1<\tau<K_R$ have the same sum-DoF as the MSIO $X$-network? In~\cite{maddah2015cache}, a delivery scheme for a $ 3\times 3$ cooperative network was proposed to achieve a sum-DoF of $\tau{K_T\choose \tau}/\left(\tau{K_T\choose \tau}+1\right)$ for $\tau=K_R-1=2$. Moreover, the authors in~\cite{xu2017fundamental} proposed a delivery scheme for cooperative $X$-network achieving sum-DoF of $\tau{K_T\choose \tau}/\left(\tau{K_T\choose \tau}+1\right)$ for $\tau=K_R-1$ for arbitrary number of transmitters and receivers. The achievability proofs in~\cite{maddah2015cache} and~\cite{xu2017fundamental} are mainly based on generating independent precoding factors at the transmitters to deal with the dependence of the derived channel coefficients. Therefore, both cooperative $X$-networks and MISO $X$-networks have the same sum-DoF of $\tau{K_T\choose \tau}/\left(\tau{K_T\choose \tau}+1\right)$ when $\tau=K_R-1$. However, is it possible to achieve the same results for cooperative $X$-networks for all cooperation orders $1<\tau<K_R$? To answer this open question, we derive an upper bound on the sum-DoF of the cooperative $X$-networks presented in the following theorem.  

\begin{theorem}~\label{Th2} The sum degrees of freedom of the cooperative $X$-network of $K_T$ transmitters, $K_R$ receivers, and $K_R{K_T\choose \tau}$ messages satisfies
\begin{equation}
\mathsf{DoF}\left(\tau\right)\leq \min_{\tau \leq \sigma\leq\min\left\{K_T,K_R\right\}}\frac{\tau{K_T\choose \tau}K_R}{\tau{K_T\choose \tau}+\left(K_R-\sigma\right){\sigma-1\choose \tau-1}} 
\end{equation}
\end{theorem} 
\begin{proof}
The proof is presented in Appendix~\ref{App1}. 
\end{proof}
In Theorem~\ref{Th2}, it is shown that the upper bound on the sum-DoF of the cooperative $X$-network coincides with the sum-DoF of the MISO $X$-network when $\tau=K_R-1$. Hence, the achievable schemes in~\cite{maddah2015cache} and~\cite{xu2017fundamental} are optimal when $\tau=K_R-1$. However, the upper bound in Theorem~\ref{Th2} shows that the cooperative $X$-network has a smaller sum-DoF than the MISO $X$-network when $1<\tau<K_R-1$. Therefore, we cannot proceed along the lines of the delivery scheme in~\cite{maddah2015cache} to achieve sum-DoF of \small$\frac{\tau {K_T\choose \tau} K_R}{\tau {K_T\choose \tau}+K_R-\tau}$ \normalsize for $1<\tau<K_R-1$ which is fundamentally different from the case of $\tau=K_R-1$.

Theorem~\ref{Th2} not only gives an upper bound on the sum-DoF of cooperative $X$-networks, but also gives an upper bound on the sum-DoF of cache-aided interference networks that uses the placement strategy introduced in Subsection~\ref{Subsec2}. In other words, there is no delivery scheme for cache-aided interference networks with transmitter caches of size $M=\tau N/K_T$ files achieving a sum-DoF higher than the bound given in Theorem~\ref{Th2} under the content placement strategy introduced in Subsection~\ref{Subsec2}. 

In~\cite{sengupta2016cache}, an information-theoretic bound on the normalized delivery time (NDT) of the cache-aided interference networks is proposed. The NDT refers to the worst case delivery time for satisfying the receivers requests. Thus, the relation between NDT and the sum-DoF is given by $\mathsf{NDT}=K_R/\mathsf{DoF}$. Figure~\ref{fig4} depicts the inverse of the upper bound on the sum-DoF proposed in~\cite[Theorem~$1$]{sengupta2016cache} for the $10\times 10$ cache-aided networks and the inverse of the upper bound on the sum-DoF of the cooperative $X$-networks in Theorem~\ref{Th2} for $M=\tau N/K_T$, $\tau\in\left[10\right]$. Note that, as demonstrated in Figure~\ref{fig4}, the upper bound on the cooperative $X$-networks is tighter than the upper bound introduced in~\cite{sengupta2016cache} for smaller cache sizes. This may be due to two possible explanations. First, it is possible that the content placement introduced in Section~\ref{Subsec2} is not sufficient to achieve the upper bound in~\cite{sengupta2016cache}. Second, the upper bound in~\cite{sengupta2016cache} is not achievable for small cache sizes. This observation motivates future work to find better placement strategies or upper bounds for cache-aided interference networks. 

\begin{figure}[t]
\centering{\includegraphics[scale=0.37,trim=4 4 4 4,clip]{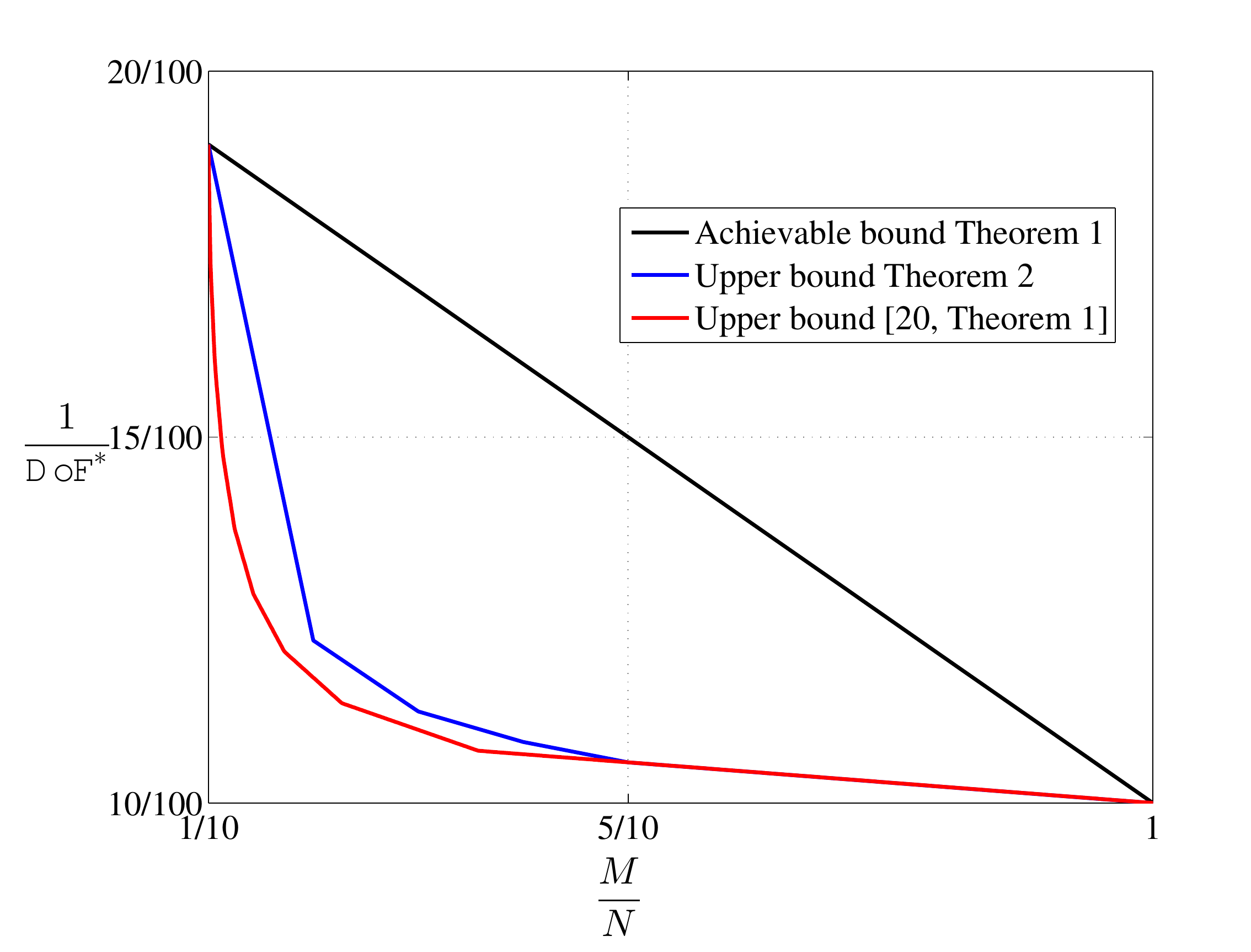}}
\caption{Comparison of the reciprocal of the upper bound on the sum-DoF proposed in~\cite{sengupta2016cache} and the reciprocal of the upper bound in Theorem~\ref{Th2} for the $10\times 10$ cache-aided interference network.}
\label{fig4}\vspace{-8mm}
\end{figure}

Before presenting the achievable scheme for the general case, we first illustrate the main idea of the achievable scheme, using an example.

\subsection{Achievable scheme for $3\times 3$ cache-Aided interference networks}\label{Subsec1}

Consider a cache-aided interference network with $K_T=3$ transmitters each with a cache $\lbrace Z_i\rbrace_{i=1}^{3}$ of size $M=2$ files, $K_R=3$ receivers, and a library of $N=3$ files: $\mathcal{W}=\left\{A,B,C\right\}$, each consisting of $F$ bits.

 \textit{In the placement phase}, each file $W_f$ is split into ${3\choose 2}=3$ subfiles of equal size $F/3$ bits. We label each subfile with $W_{f,\mathcal{S}}$, where $\mathcal{S}\subset\left[3\right]$ and $|\mathcal{S}|=2$ such that the subfile $W_{f,\mathcal{S}}$ is stored at transmitters $\text{TX}_i$, $i\in\mathcal{S}$. For example, file $A$ is split as follows $A=\left(A_{12},A_{13},A_{23}\right)$, where subfile $A_{12}$ is stored at transmitters $\text{TX}_1$ and $\text{TX}_2$, subfile $A_{13}$ at $\text{TX}_1$ and $\text{TX}_3$, and subfile $A_{23}$ at $\text{TX}_2$ and $\text{TX}_3$. Hence, the cache contents at each transmitter are given by
\begin{equation*}
\begin{aligned}
Z_1=\left(A_{12},A_{13},B_{12},B_{13},C_{12},C_{13}\right)\\
Z_2=\left(A_{12},A_{23},B_{12},B_{23},C_{12},C_{23}\right)\\
Z_3=\left(A_{13},A_{23},B_{13},B_{23},C_{13},C_{23}\right)
\end{aligned}
\end{equation*}  
Hence, each transmitter stores $3*{3-1\choose 2-1}=6$ subfiles with each of it having a size of $F/3$ bits resulting in a total of $2F$ bits that satisfies the memory size constraint. 

\textit{In the delivery phase}, we assume that the receiver demand files $W_{d_1}=A$, $W_{d_2}=B$, and $W_{d_3}=C$. Consider subfile $W_{f,\lbrace ij\rbrace}$, $i<j$, $\left(i,j\right)\in\left[3\right]$ stored at two transmitters $\text{TX}_i$ and $\text{TX}_j$. We split subfile $W_{f,\lbrace ij\rbrace}$ into two disjoint smaller subfiles of equal sizes $W_{f,\lbrace ij\rbrace}^{i},W_{f,\lbrace ij\rbrace}^{j}$. For subfile $ W_{f,\lbrace ij\rbrace}^{i}$, the transmitter $\text{TX}_i$ is the primary (master) transmitter responsible for delivering this subfile, while transmitter $\text{TX}_j$ represents the cooperative transmitter that has this subfile as a side information. For example, subfile $A_{13}$ is split into $A_{13}=\left(A_{13}^{1},A_{13}^{3}\right)$, where subfile $A_{13}^{1}$ is assigned to transmitter $\text{TX}_1$ while it is also available at transmitter $\text{TX}_3$ as a side information.

 Consequently, we schedule the $18$ small subfiles into $2$ sets with each set containing $9$ subfiles, such that the subfiles in the same set are delivered simultaneously in the same transmission block. In other words, it requires $2$ transmission blocks to complete the transmission, where each transmission block occurs over $\mu_n$-symbol extensions. The value of $\mu_n$ will be specified later. Figure~\ref{fig3} illustrates the scheduling strategy and the set of subfiles transmitted in each transmission block. In each transmission block, each transmitter has a distinct subfile intended to each receiver, and also has the subfiles of an adjacent transmitter as a side information. For example, consider the first transmission block (Fig.~\ref{fig3_A}). Transmitter $\text{TX}_1$ has subfiles $A_{12}^{1}$, $B_{12}^{1}$, and $C_{12}^{1}$ intended to receivers $\text{RX}_1$, $\text{RX}_2$, and $\text{RX}_3$, respectively. In addition, the subfiles $A_{13}^{3}, B_{13}^{3}, C_{13}^{3}$ are available as a side information at transmitter $\text{TX}_1$. Thus, we consider transmitter $\text{TX}_1$ as a cooperative (cognitive) transmitter to the transmitter $\text{TX}_3$. 

We now present the delivery scheme at the first transmission block. The delivery scheme is based on designing two-layer beamforming matrices at the transmitters, where the first layer is responsible for applying zero-forcing and the second layer is responsible for applying asymptotic interference alignment. For transmitter $\text{TX}_1$, the subfiles $A_{12}^{1}$, $B_{12}^{1}$, $C_{12}^{1}$ are transmitted using beamforming matrices $\mathbf{V}_{a,12}^{1}$, $\mathbf{V}_{b,12}^{1}$, $\mathbf{V}_{c,12}^{1}$, respectively, and subfiles $A_{13}^{3}, B_{13}^{3}, C_{13}^{3}$ are transmitted using beamforming matrices $\mathbf{V}_{a,13}^{1}$, $\mathbf{V}_{b,13}^{1}$, $\mathbf{V}_{c,13}^{1}$, respectively. Each beamforing matrix consists of two layers, e.g., $\mathbf{V}_{a,12}^{1}=\mathbf{V}_{a,12}^{1,\text{ZF}}\mathbf{V}_{a}^{\text{IA}}$, where the first layer precoder $\mathbf{V}_{a,12}^{1,\text{ZF}}$ is different for each transmitter while the second precoder $\mathbf{V}_{a}^{\text{IA}}$ is the same for each transmitter.

\begin{figure}[t]
\centering
\begin{subfigure}[b]{0.48\linewidth}
  \centerline{\includegraphics[scale=0.38,trim=4 4 4 4,clip]{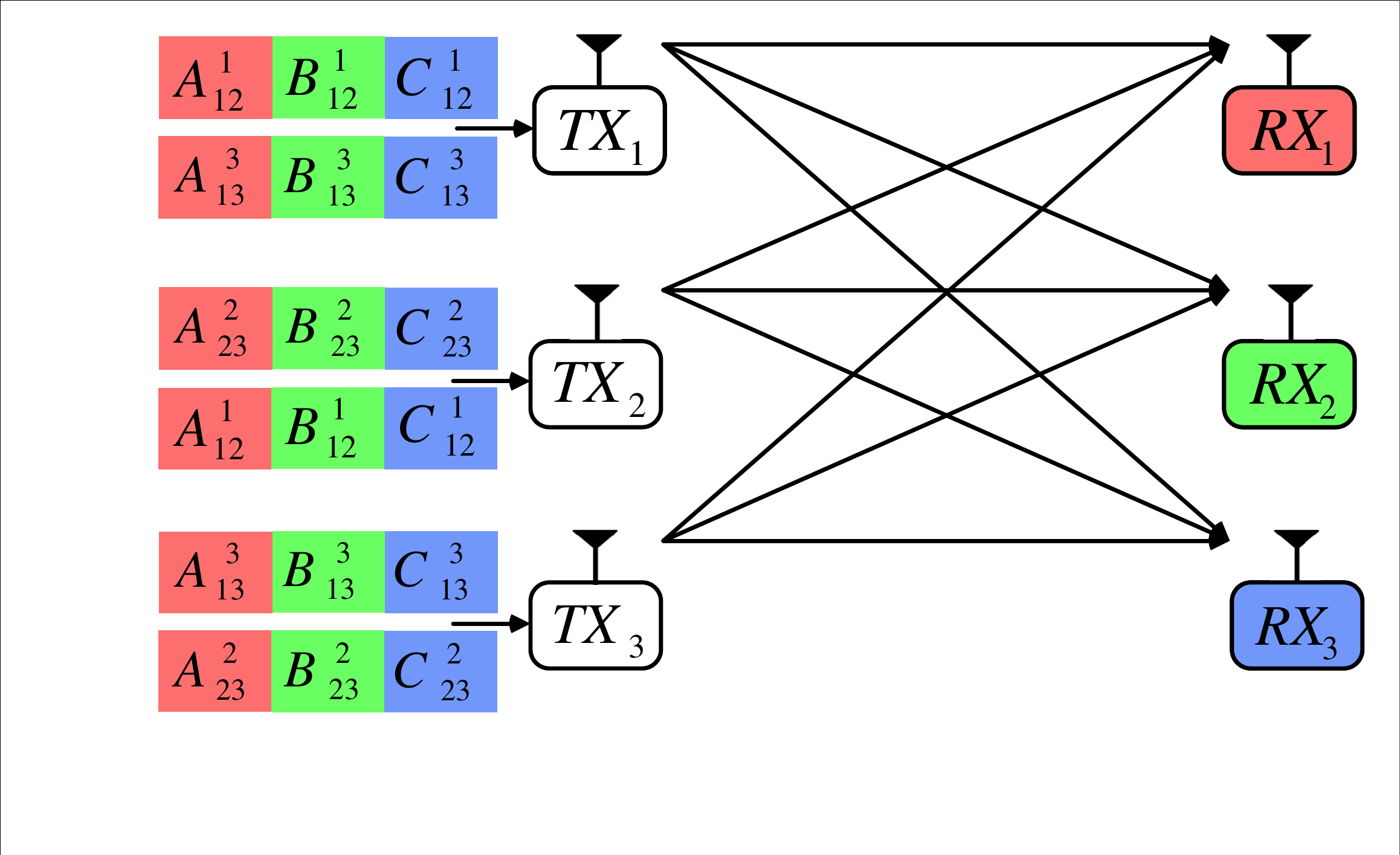}}
 \caption{First transmission block}\label{fig3_A}
 \end{subfigure}
 \begin{subfigure}[b]{0.48\linewidth}
  \centerline{\includegraphics[scale=0.38,trim=4 4 4 4,clip]{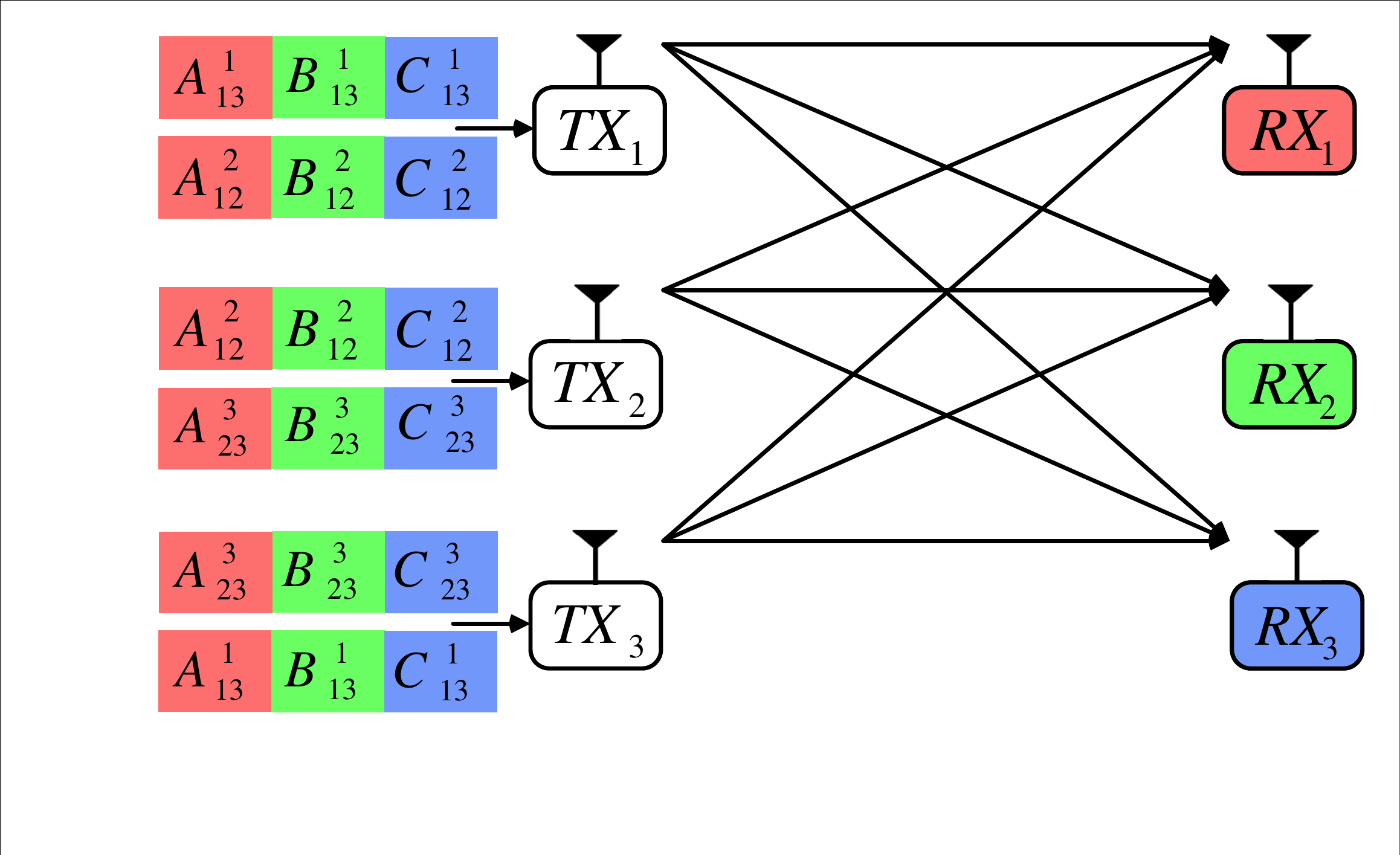}}
\caption{Second transmission block}\label{fig3_B}
 \end{subfigure}
\caption{Description of the delivery scheme for the $3\times 3$ cache-aided interference network with a library of $3$ files and transmitter cache size $M=2$ files.}
\label{fig3}\vspace{-7mm}
\end{figure}

 The first precoder is designed such that the messages intended to receiver $\text{RX}_i$ from all transmitters are canceled at the adjacent receiver $\text{RX}_{i+1}$. In other words, the objective of the first layer is to zero-force the subfiles $\left(A_{12}^{1},A_{23}^{2},A_{13}^{3}\right)$ at receiver $\text{RX}_2$, subfiles $\left(B_{12}^{1},B_{23}^{2},B_{13}^{3}\right)$ at receiver $\text{RX}_3$, and $\left(C_{12}^{1},C_{23}^{2},C_{13}^{3}\right)$ at receiver $\text{RX}_1$. For example, to cancel the interference of subfile $A_{12}^{1}$ at receiver $\text{RX}_2$, we set precoders $\mathbf{V}_{a,12}^{1,\text{ZF}}$ and $\mathbf{V}_{a,12}^{2,\text{ZF}}$ as follows

\begin{equation}
\begin{aligned}
&V_{a,12}^{1,\text{ZF}}\left(u\right)=h_{22}\left(u\right),\ \forall u\in\left[\mu_n\right]\\
&V_{a,12}^{2,\text{ZF}}\left(u\right)=-h_{21}\left(u\right),\ \forall u\in\left[\mu_n\right],
\end{aligned}
\end{equation}
where $V_{a,12}^{1,\text{ZF}}\left(u\right)$, $V_{a,12}^{2,\text{ZF}}\left(u\right)$ are the $u$-th diagonal element of matrices $\mathbf{V}_{a,12}^{1,\text{ZF}}$, $\mathbf{V}_{a,12}^{2,\text{ZF}}$, respectively. Thus, subfile $A_{12}^{1}$ is received with directions $\mathbf{M}_{12}^{12}\mathbf{V}_{a}^{\text{IA}}$, $\mathbf{0}$, $\mathbf{M}_{23}^{12}\mathbf{V}_{a}^{\text{IA}}$ at receivers $\text{RX}_1$, $\text{RX}_2$, and $\text{RX}_3$, respectively, where $\mathbf{M}_{gk}^{lm}$ is a $\mu_n\times \mu_n$ diagonal matrix with the $u$-th element given by
 \begin{equation}
 M_{gk}^{lm}\left(u\right)=\mathit{det}\begin{bmatrix}
 h_{gl}\left(u\right)&h_{gm}\left(u\right)\\
 h_{kl}\left(u\right)&h_{km}\left(u\right)\\
 \end{bmatrix}.
 \end{equation}
The received directions of all subfiles after the first layer of precoding at all receivers is shown in Table~\ref{Tab1}. Thus, each receiver has three desired signals and three interference signals. Therefore, we design the second layer of precoders $\mathbf{V}_{a}^{\text{IA}}$, $\mathbf{V}_{b}^{\text{IA}}$, $\mathbf{V}_{c}^{\text{IA}}$ such that the interference signals at receiver $\text{RX}_1$ are aligned at a space of dimension $|\mathbf{V}_{b}^{\text{IA}}|$. Similarly, the interference signals at receiver $\text{RX}_2$ are aligned at a space of dimension $|\mathbf{V}_{c}^{\text{IA}}|$, and the interference signals at receiver $\text{RX}_3$ are aligned at a space of dimension $|\mathbf{V}_{a}^{\text{IA}}|$. We can satisfy these alignment conditions by applying asymptotic interference alignment scheme as in~\cite{cadambe2009interference}. The main challenge here is that the derived channel coefficients are polynomial functions of the original channel coefficients. Thus, we prove that the interference alignment conditions are still feasible in our scenario in Appendix~\ref{App4}. For symmetry, we set $|\mathbf{V}_{a}^{\text{IA}}|=|\mathbf{V}_{b}^{\text{IA}}|=|\mathbf{V}_{c}^{\text{IA}}|=|\mathbf{V}|$. As a result, at each receiver, the desired signals occupy a space of dimension $3|\mathbf{V}|$ and after applying the alignment strategy, the interference occupies a space of dimension $|\mathbf{V}|$. Hence, if we consider $\mu_n\approx 4|\mathbf{V}|$ symbol extensions such that each desired signal spans a space linearly independently of the other desired signals and the interference spaces, then each receiver achieves a DoF of $\frac{3|\mathbf{V}|}{4|\mathbf{V}|}$. Thus, a total of $9/4$ sum-DoF is achievable in the first transmission block. In a similar manner, we repeat the same scheme in the second transmission block. Therefore, the sum-DoF of $9/4$ is achievable, overall.
  
\begin{table}[t]
\caption{The received directions of the subfiles at each receiver after the first layer of precoding}
\begin{center}
\begin{tabular}{ |c|c|c|c| }
\hline
\multirow{2}{*}{Subfiles}&\multicolumn{3}{ |c| }{Received directions}\\ \cline{2-4}
&$\text{RX}_1$ &$\text{RX}_2$&$\text{RX}_3$\\ 
\hline
$A_{12}^{1}$&$\mathbf{M}_{12}^{12}\mathbf{V}_{a}^{\text{IA}}$&$\mathbf{0}$&$\mathbf{M}_{23}^{12}\mathbf{V}_{a}^{\text{IA}}$\\
\hline
$A_{23}^{2}$&$\mathbf{M}_{12}^{23}\mathbf{V}_{a}^{\text{IA}}$&$\mathbf{0}$&$\mathbf{M}_{23}^{23}\mathbf{V}_{a}^{\text{IA}}$\\
\hline
$A_{13}^{3}$&$\mathbf{M}_{12}^{13}\mathbf{V}_{a}^{\text{IA}}$&$\mathbf{0}$&$\mathbf{M}_{23}^{13}\mathbf{V}_{a}^{\text{IA}}$\\
\hline
$B_{12}^{1}$&$\mathbf{M}_{13}^{12}\mathbf{V}_{b}^{\text{IA}}$&$\mathbf{M}_{23}^{12}\mathbf{V}_{b}^{\text{IA}}$&$\mathbf{0}$\\
\hline
$B_{23}^{2}$&$\mathbf{M}_{13}^{23}\mathbf{V}_{b}^{\text{IA}}$&$\mathbf{M}_{23}^{23}\mathbf{V}_{b}^{\text{IA}}$&$\mathbf{0}$\\
\hline
$B_{13}^{3}$&$\mathbf{M}_{13}^{13}\mathbf{V}_{b}^{\text{IA}}$&$\mathbf{M}_{23}^{13}\mathbf{V}_{b}^{\text{IA}}$&$\mathbf{0}$\\
\hline
$C_{12}^{1}$&$\mathbf{0}$&$\mathbf{M}_{12}^{12}\mathbf{V}_{c}^{\text{IA}}$&$\mathbf{M}_{13}^{12}\mathbf{V}_{c}^{\text{IA}}$\\
\hline
$C_{23}^{2}$&$\mathbf{0}$&$\mathbf{M}_{12}^{23}\mathbf{V}_{c}^{\text{IA}}$&$\mathbf{M}_{13}^{23}\mathbf{V}_{c}^{\text{IA}}$\\
\hline
$C_{13}^{3}$&$\mathbf{0}$&$\mathbf{M}_{12}^{13}\mathbf{V}_{c}^{\text{IA}}$&$\mathbf{M}_{13}^{13}\mathbf{V}_{c}^{\text{IA}}$\\
\hline
\end{tabular}
\end{center}
\label{Tab1}
\vspace{-10mm}
\end{table}

In the rest of the paper, we focus on proving the achievability of the lower bound on the sum-DoF in Theorem~\ref{Th1} for corner points $M=\tau N/K_T$ with $\tau\in\left[ K_T\right]$.

\section{Achievable Scheme}\label{Sec3}

In this section, we propose a new delivery scheme that achieves the lower bound in Theorem~\ref{Th1} for corner points of transmitter cache size $M=\tau N/K_T$ for $\tau\in\left[K_T\right]$. Without loss of generality, consider a demand vector $\mathbf{d}$ in which the receiver $\text{RX}_j$ requests file $W_{d_j}$ under the content placement strategy introduced in Subsection~\ref{Subsec2}. Thus, the transmitters should send the requested subfiles $\left\{W_{d_j,\mathcal{S}}:j\in\left[K_R\right],\mathcal{S}\subseteq\left[K_T\right],|\mathcal{S}|=\tau\right\}$ which means a total of $K_R{K_T\choose \tau}$ subfiles are delivered to the receivers.

\subsection{When $\tau\geq K_R$}
  The transmission is performed over ${K_T\choose \tau}$ time slots. At each slot, the set $\mathcal{S}$ of $\tau$ transmitters cooperates to send subfiles $\left\{W_{d_j,\mathcal{S}}\right\}_{j=1}^{K_R}$ to the receivers, where the network becomes as a MISO broadcast channel with a single transmitter of $\tau$ antennas and $K_R$ single-antenna receivers. Thus, beamforming and zero-forcing are sufficient to achieve $\mathsf{DoF}=K_R$, where the sum-DoF is bounded by the number of receivers $K_R$ (See Appendix~\ref{App1}).
  
   In the following, for a given set $\mathcal{G}$, we use $\Pi_{\mathcal{G}}$ to denote the set of $|\mathcal{G}|!$ permutations of set $\mathcal{G}$, and $\Pi_{\mathcal{G}}^{\text{circ}}$ to denote the set of $\left(|\mathcal{G}|-1\right)!$ circular permutations\footnote{A circular permutation of a set is an ordered selection of the set elements arranged along a fixed circle. For example, for the set $\mathcal{G}=\left\{1,2,3\right\}$, $\Pi_{\mathcal{G}}^{\text{circ}}=\left\{\left[1,2,3\right],\left[1,3,2\right]\right\}$. } of set $\mathcal{G}$. For a permutation $\pi\in\Pi_{\mathcal{G}}$, we use $\pi\left(i\right)=\pi\left(i+|\mathcal{G}|\right)$ to denote the $i$th element in permutation $\pi$, and $\pi\left[i:j\right]=\left[\pi\left(i\right),\ldots,\pi\left(j\right)\right]$ for $j\geq i$.

\subsection{When $\tau<K_R$}

 First, we split each subfile $W_{d_j,\mathcal{S}}$ for $j\in\left[K_R\right]$, $\mathcal{S}\subseteq\left[K_T\right]$ such that $|\mathcal{S}|=\tau$ into $\tau !\,\left(K_T-\tau\right)!$ disjoint smaller subfiles as follows
\begin{equation}~\label{eqn13}
W_{d_j,\mathcal{S}}\triangleq\left\{W_{d_j,\pi,\tilde{\pi}}^{\pi\left(1\right)}:\pi\in\Pi_{\mathcal{S}},\tilde{\pi}\in\Pi_{\left[K_T\right]\setminus\mathcal{S}}\right\},
\end{equation}   
where the index $\pi\left(1\right)$ of subfile $W_{d_j,\pi,\tilde{\pi}}^{\pi\left(1\right)}$ refers to the primary (master) transmitter $\text{TX}_{\pi\left(1\right)}$ responsible for delivering this subfile, while transmitters $\pi\left[2:\tau\right]$ represent the cooperative transmitters that have this subfile as a side information. Furthermore, subfile $W_{d_j,\pi,\tilde{\pi}}^{\pi\left(1\right)}$ is not available at the cache of transmitters $\text{TX}_i$, $i\in\tilde{\pi}$. This splitting strategy helps us to schedule the smaller subfiles required to be delivered to the receivers into $\left(K_T-1\right)!$ sets, where each set has $K_RK_T$ smaller subfiles that would be delivered in the same transmission block.

\begin{lemma}~\label{Lem1} The set of smaller subfiles needed to be delivered to the receivers can be partitioned into $\left(K_T-1\right)!$ disjoint sets each of size $K_RK_T$ smaller subfiles as follows
\begin{equation}
\bigcup_{\substack{\pi\in\Pi^{\text{circ}}_{\left[K_T\right]}}}\left\{W_{d_j,\pi\left[i:i+\tau-1\right],\pi\left[i+\tau:K_T+i-1\right]}^{\pi\left(i\right)}:i\in\left[K_T\right],j\in\left[K_R\right]\right\}.
\end{equation}
\end{lemma}
\begin{proof}
The proof is presented in Appendix~\ref{App3}
\end{proof}

According to Lemma~\ref{Lem1}, we have $\left( K_T-1\right)!$ sets, where each set is defined by a unique permutation $\pi\in\Pi_{\left[K_T\right]}^{\text{circ}}$, and each set contains $K_R K_T$ smaller subfiles that we call messages from now. For a given $\pi\in\Pi_{\left[K_T\right]}^{\text{circ}}$, we sort the transmitters in the permutation order $\pi$. Thus, each transmitter $\text{TX}_{\pi\left(i\right)}$, $i\in\left[K_T\right]$, has a dedicated message intended to every receiver $\text{RX}_j$, $j\in\left[K_R\right]$. Furthermore, the messages assigned to transmitter $\text{TX}_{\pi\left(i\right)}$, $\lbrace W_{d_j,\pi\left[i:i+\tau-1\right],\pi\left[i+\tau:K_T+i-1\right]}^{\pi\left(i\right)}\rbrace_{j=1}^{K_R}$, are available at the adjacent $\tau-1$ cooperative transmitters $\lbrace \text{TX}_{l}:l\in\pi\left[i+1:i+\tau-1\right]\rbrace$, while these messages are not available at the remaining $\left( K_T-\tau\right)$ transmitters $\text{TX}_l$, $l\in\pi\left[i+\tau:K_T+i-1\right]$.
 For any arbitrary $n\in\mathbb{N}^{+}$, we consider each transmission block occurs over $\mu_n=K_T n^{\Gamma}+\left(K_R-\tau\right)\left(n+1\right)^{\Gamma}$ symbol extensions of the original channel, where $\Gamma=K_T\left(K_R-\tau\right)$. We design a coding scheme to deliver the messages in each block with a sum-DoF of \small$\frac{K_RK_Tn^{\Gamma}}{K_Tn^{\Gamma}\left(K_R-\tau\right)\left(n+1\right)^{\Gamma}}$. \normalsize Note that as $n\to\infty$, we obtain $\mathsf{DoF}=\frac{K_RK_T}{K_R+K_T-\tau}$. The main idea behind the proposed coding scheme in each transmission block is that each transmitter has $K_R$ messages, where each message is required by a dedicated receiver, and these messages are available at $\tau-1$ cooperative transmitters as a side information. Thus, we develop two-layer beamforming matrices at each transmitter. In the first layer, each transmitter with the help of the cooperative transmitters nulls out the message intended to a specific receiver into the neighboring $\tau-1$ receivers. Hence, for every receiver, there are $K_T$ desired messages, one from each transmitter, in addition to $K_T\left(K_R-\tau\right)$ interference messages. Then, we design an asymptotic interference alignment scheme in the second layer to let the interference signals span a $\left( K_R-\tau\right)$ dimensional space by aligning every $K_T$ interference signals into a single-dimensional space. Therefore, every receiver can decode the desired $K_T$ messages over $K_T+K_R-\tau$ dimensional spaces to get a DoF of $\frac{K_T}{K_T+K_R-\tau}$ per receiver.

Without loss of generality, let us focus on the transmission block for $\pi=\left[1,\ldots,K_T\right]$ for delivering subfiles $\left\{W_{d_j,\left[i:i+\tau-1\right],\left[i+\tau:K_T+i-1\right]}^{i}:i\in\left[K_T\right],j\in\left[K_R\right]\right\}$, 
where the delivery scheme of the remaining blocks is performed in the same manner. The input-output relation of the original channel in~\eqref{eqn1} over $\mu_n$-symbol extensions is given by
\begin{equation}
\mathbf{Y}_k=\sum_{i=1}^{K_T} \mathbf{H}_{ki}\mathbf{X}_i+\mathbf{N}_{k},
\end{equation}
where $\mathbf{Y}_k$ and $\mathbf{N}_{k}$ represent $\mu_n\times 1$ column vectors of the received signal and the Gaussian noise of the receiver $\text{RX}_k$ over $\mu_n$-symbol extension, respectively. $\mathbf{X}_i$ is a $\mu_n\times 1$ column vector representing the transmitted vector of transmitter $\text{TX}_i$. $\mathbf{H}_{ki}$ is a $\mu_n\times \mu_n$ diagonal matrix of channel coefficients from transmitter $\text{TX}_i$ to receiver $\text{RX}_k$ over $\mu_n$ symbol extension.
\begin{equation}
\mathbf{H}_{ki}=\begin{bmatrix}
h_{ki}\left(1\right)& 0 & \dots & 0\\
0 & h_{ki}\left(2\right) &\dots & 0\\
\vdots & \vdots & \ddots & \vdots\\
0& 0& \dots &h_{ki}\left(\mu_n\right) 
\end{bmatrix}.
\end{equation}

Each message $W_{d_j,\left[i:i+\tau-1\right],\left[i+\tau:K_T+i-1\right]}^{i}$, $i\in\left[K_T\right]$, $j\in\left[K_R\right]$, is encoded into $n^{\Gamma}$ independent streams represented by a $n^{\Gamma} \times 1$ column vector $\mathbf{X}_{j,\left[i:i+\tau-1\right]}$. Note that each vector $\mathbf{X}_{j,\left[i:i+\tau-1\right]}$ can be constructed at transmitters $\text{TX}_l$, $l\in\left[i:i+\tau-1\right]$, i.e., it can be sent from $\tau$ transmitters. Hence, transmitter $\text{TX}_l$ uses beamforming  matrix $\mathbf{V}_{j,\left[i:i+\tau-1\right]}^{l}$ for precoding the vector $\mathbf{X}_{j,\left[i:i+\tau-1\right]}$, where $\mathbf{V}_{j,\left[i:i+\tau-1\right]}^{l}$ is a $\mu_n\times n^{\Gamma}$ matrix. Therefore, we can describe the transmitted vector of transmitter $\text{TX}_{i}$ as follows:
\begin{equation}
\mathbf{X}_{i}=\sum_{j=1}^{K_R}\sum_{l=i-\tau+1}^{i} \mathbf{V}_{j,\left[l:l+\tau-1\right]}^{i} \mathbf{X}_{j,\left[l:l+\tau-1\right]},
\end{equation}   
where it can be verified that transmitter $\text{TX}_{i}$ is a cooperative transmitter for transmitters $l\in\left[i-\tau+1:i\right]$, i.e., transmitter $\text{TX}_{i}$ has the messages of transmitters $\text{TX}_l$, $l\in\left[i-\tau+1:i\right]$, as a side information. Furthermore, the received signal at receiver $\text{RX}_k$ is given by
\begin{equation}
\begin{aligned}
\mathbf{Y}_k&=  \sum_{i=1}^{K_T}\mathbf{H}_{ki}\left(\sum_{j=1}^{K_R}\sum_{l=i-\tau+1}^{i} \mathbf{V}_{j,\left[l:l+\tau-1\right]}^{i} \mathbf{X}_{j,\left[l:l+\tau-1\right]}\right) +\mathbf{N}_k\\
  &=  \sum_{j=1}^{K_R}\sum_{i=1}^{K_T}\sum_{l=i}^{i+\tau-1}\mathbf{H}_{kl}\mathbf{V}_{j,\left[i:i+\tau-1\right]}^{l} \mathbf{X}_{j,\left[i:i+\tau-1\right]}+\mathbf{N}_k.
\end{aligned}
\end{equation}
Note that the received signal spans a space of dimension $\mu_n=K_T n^{\Gamma}+\left(K_R-\tau\right)\left(n+1\right)^{\Gamma}$. Our objective is to design the beamforming matrices $\left\{\mathbf{V}_{j,\left[i:i+\tau-1\right]}^{l}\right\}$ such that the interference signals occupy a subspace of dimension $\left(K_R-\tau\right)\left(n+1\right)^{\Gamma}$, leaving the desired signals to span an interference-free subspace of dimension $K_Tn^{\Gamma}$ out of $\mu_n$ dimensions. We propose the beamforming matrix $\mathbf{V}_{j,\left[i:i+\tau-1\right]}^{l}$ comprising of two layers (two precoder matrices) as follows
\begin{equation}~\label{eqn2}
\mathbf{V}_{j,\left[i:i+\tau-1\right]}^{l}=\mathbf{V}_{j,\left[i:i+\tau-1\right]}^{l,\text{ZF}} \mathbf{V}_{j,\left[i:i+\tau-1\right]}^{\text{IA}},
\end{equation}
where $\mathbf{V}_{j,\left[i:i+\tau-1\right]}^{l,\text{ZF}}$ is a $\mu_n\times\mu_n$ diagonal matrix designed based on zero-forcing technique. The objective of these precoders is to null out the signals $\lbrace\mathbf{X}_{j,\left[i:i+\tau-1\right]}\rbrace_{i=1}^{K_T}$, intended to receiver $\text{RX}_j$ from all transmitters, at the neighboring $\tau-1$ undesired receivers; however, these signals still cause interference at the remaining $\left( K_R-\tau\right)$ receivers. Observe that the precoders $\lbrace\mathbf{V}_{j,\left[i:i+\tau-1\right]}^{l,\text{ZF}}\rbrace$ are different at each transmitter $\text{TX}_l$, $l\in\left[i:i+\tau-1\right]$, i.e., $\mathbf{V}_{j,\left[i:i+\tau-1\right]}^{l,\text{ZF}}\neq \mathbf{V}_{j,\left[i:i+\tau-1\right]}^{\tilde{l},\text{ZF}}$ for $\left(l,\tilde{l}\right)\in\left[i:i+\tau-1\right]$ and $l\neq \tilde{l}$. The second precoder $\mathbf{V}_{j,\left[i:i+\tau-1\right]}^{\text{IA}}$ is a $\mu_n\times n^{\Gamma}$ matrix designed to apply asymptotic interference alignment such that the signals intended to receiver $\text{RX}_j$ from all transmitters are aligned at $\left(K_R-\tau\right)$ undesired receivers into a single space of dimension  $\left(n+1\right)^{\Gamma}$, where this precoder is the same at all transmitters $\text{TX}_l$, $l\in\left[i:i+\tau-1\right]$. In the following, we give the design of these two-layers of precoders in detail.

\subsubsection{Design of the Zero-forcing Precoder} 

Given the beamforming matrix design in~\eqref{eqn2}, we can rewrite the received signal at receiver $\text{RX}_k$ as
\begin{equation}
\mathbf{Y}_k=\sum_{j=1}^{K_R}\sum_{i=1}^{K_T}\left(\sum_{l=i}^{i+\tau-1}\mathbf{H}_{kl}\mathbf{V}_{j,\left[i:i+\tau-1\right]}^{l,\text{ZF}}\right) \mathbf{V}_{j,\left[i:i+\tau-1\right]}^{\text{IA}}\mathbf{X}_{j,\left[i:i+\tau-1\right]}+\mathbf{N}_k.
\end{equation}
Observe that the signal $\mathbf{X}_{j,\left[i:i+\tau-1\right]}$ intended to receiver $\text{RX}_j$ is available at $\tau$ transmitters, and hence, these transmitters can cooperate with each other to null out this signal at $\tau-1$ unintended receivers. Thus, we aim to design precoders $\lbrace\mathbf{V}_{j,\left[i:i+\tau-1\right]}^{l,\text{ZF}} \rbrace$ such that the signal $\mathbf{X}_{j,\left[i:i+\tau-1\right]}$ does not cause interference at receivers $\text{RX}_k$, $k\in\left[j+1:j+\tau-1\right]$. For given $i\in\left[K_T\right]$, it is required for precoders $\lbrace\mathbf{V}_{j,\left[i:i+\tau-1\right]}^{l,\text{ZF}}\rbrace$, $l\in\left[i:i+\tau-1\right]$, to satisfy the conditions 
 
 \small
\begin{equation}~\label{eqn3}
\sum_{l=i}^{i+\tau-1} \mathbf{H}_{kl} \mathbf{V}_{j,\left[i:i+\tau-1\right]}^{l,\text{ZF}}=\begin{bmatrix}
\sum_{\substack{l=i}}^{i+\tau-1}h_{kl}\left(1\right)V_{j,\left[i:i+\tau-1\right]}^{l,\text{ZF}}\left(1\right)&\dots&0\\
\vdots&\ddots&\vdots\\
0&\dots&\sum_{l=i}^{i+\tau-1}h_{kl}\left(\mu_n\right)V_{j,\left[i:i+\tau-1\right]}^{l,\text{ZF}}\left(\mu_n\right)
\end{bmatrix}=\mathbf{0}_{\mu_n\times \mu_n}
\end{equation}
\normalsize
for all $k\in\left[j+1:j+\tau-1\right]$, where $V_{j,\left[i:i+\tau-1\right]}^{l,\text{ZF}}\left(u\right)$ is the $u$-th diagonal element of the precoder matrix $\mathbf{V}_{j,\left[i:i+\tau-1\right]}^{l,\text{ZF}}$. Let $\mathbf{H}\left(u\right)$ denote the $K_R\times K_T$ channel matrix between $K_T$ transmitters and $K_R$ receivers at time slot $u\in\left[\mu_n\right]$. Moreover, $\mathbf{H}_{\mathcal{S}_R}^{\mathcal{S}_T}\left(u\right)$ represents $|\mathcal{S}_R|\times |\mathcal{S}_T|$ submatrix of the channel matrix $\mathbf{H}\left(u\right)$ formed by taking rows indexed by $\mathcal{S}_R$ and columns indexed by $\mathcal{S}_T$. For arbitrary vector $\mathbf{a}=\left(a_1,\cdots,a_\tau\right)$, we consider the determinant of the following matrix using cofactor expansion
\begin{equation}
\mathit{det}\begin{bmatrix}
h_{j+1,i}\left(u\right)&\dots & h_{j+1,i+\tau-1}\left(u\right)\\
\vdots&&\vdots\\
h_{j+\tau-1,i}\left(u\right)&\dots & h_{j+\tau-1,i+\tau-1}\left(u\right)\\
a_1&\dots&a_{\tau}
\end{bmatrix}=\sum_{l=1}^{\tau}a_l c_{l}\left(u\right),
\end{equation} 
where $c_l\left(u\right)$ is the cofactor of element $a_l$. By taking $V_{j,\left[i:i+\tau-1\right]}^{l,\text{ZF}}\left(u\right)=c_l\left(u\right)$, we can verify that
\begin{equation}
\begin{aligned}
\sum_{l=i}^{i+\tau-1}h_{kl}\left(u\right)V_{j,\left[i:i+\tau-1\right]}^{l,\text{ZF}}\left(u\right)&=\mathit{det}\begin{bmatrix}
h_{j+1,i}\left(u\right)&\dots & h_{j+1,i+\tau-1}\left(u\right)\\
\vdots&&\vdots\\
h_{j+\tau-1,i}\left(u\right)&\dots & h_{j+\tau-1,i+\tau-1}\left(u\right)\\
h_{ki}\left(u\right)&\dots&h_{k,i+\tau-1}\left(u\right)
\end{bmatrix}\\
&=\left\{\begin{array}{ll}
0 & \text{if}\ k\in\left[j+1:j+\tau-1\right]\\
M_{\lbrace k\rbrace\cup\left[j+1:j+\tau-1\right]}^{\left[i:i+\tau-1\right]}\left(u\right)& \text{otherwise}
\end{array}\right.
\end{aligned}
\end{equation}
where $M_{\lbrace k\rbrace\cup\left[j+1:j+\tau-1\right]}^{\left[i:i+\tau-1\right]}\left(u\right)$ is the determinant of the submatrix $\mathbf{H}_{\lbrace k\rbrace\cup\left[j+1:j+\tau-1\right]}^{\left[i:i+\tau-1\right]}\left(u\right)$. Therefore, the conditions in~\eqref{eqn3} are satisfied. Thus, the received signal at receiver $\text{RX}_k$ can be rewritten as
\begin{equation}~\label{eqn4}
\mathbf{Y}_k=\sum_{i=1}^{K_T}\mathbf{M}_{\left[ k:k+\tau-1\right]}^{\left[i:i+\tau-1\right]}\mathbf{V}_{k,\left[i:i+\tau-1\right]}^{\text{IA}}+\sum_{i=1}^{K_T}\sum_{j=k+1}^{K_R+k-\tau}\mathbf{M}_{\lbrace k\rbrace \cup\left[j+1:j+\tau-1\right]}^{\left[i:i+\tau-1\right]}\mathbf{V}_{j,\left[i:i+\tau-1\right]}^{\text{IA}} + \mathbf{N}_k,
\end{equation}
where $\mathbf{M}_{\lbrace k\rbrace \cup\left[j+1:j+\tau-1\right]}$ is a $\mu_n\times \mu_n$ diagonal matrix with the $u$-th diagonal element $M_{\lbrace k\rbrace \cup\left[j+1:j+\tau-1\right]}^{\left[i:i+\tau-1\right]}\left(u\right)$ for $u\in\left[\mu_n\right]$. As a result, the signals intended to receiver $\text{RX}_j$ from all transmitters are canceled at receivers $k\in\left[j+1:j+\tau-1\right]$, while these signals interfere at receivers $\text{RX}_k$, $k\in\left[j+\tau:K_R+j-1\right]$.
\begin{remark}(\textit{Full rank of the first-layer precoder}): Note that the entries of the channel matrix $\mathbf{H}\left(u\right)$ are drawn i.i.d. from a continuous distribution. Hence, $\mathbf{H}\left(u\right)$ has full rank with probability one for all $u\in\left[\mu_n\right]$. Furthermore, any square submatrix from $\mathbf{H}\left(u\right)$ also has full rank with probability one. Thus, cofactors $\lbrace c_l\left(u\right)\rbrace_{u=1}^{\mu_n}$ have non-zero values, and they are independent of each other, since they are formed from i.i.d. channel coefficients. Therefore, the first-layer precoder $\mathbf{V}_{j,\left[i:i+\tau-1\right]}^{l,\text{ZF}}$ is a full rank matrix almost surely.
\end{remark}

\subsubsection{Design of the Interference Alignment Precoder}
 The received signal at receiver $\text{RX}_k$ given in~\eqref{eqn4} consists of two terms: the first term represents $K_T$ desired data streams, and the second term is the interference signals from the data streams intended to receivers $\text{RX}_j$, $j\in\left[k+1:K_R+k-\tau\right]$. Our objective is to align the interference signals at receiver $\text{RX}_k$ into a subspace of dimension $\left(K_R-\tau\right)\left(n+1\right)^{\Gamma}$ in order to let the remaining $\mu_n-\left(K_R-\tau\right)\left(n+1\right)^{\Gamma}=K_Tn^{\Gamma}$ dimensions to the desired signals. Consider the signals desired by receiver $\text{RX}_j$ that causes interference at $\left( K_R-\tau\right)$ receivers $\text{RX}_k$, $k\in\left[j+\tau:K_R+j-1\right)]$. We would like to align these signals into a single subspace of dimension $\left(n+1\right)^{\Gamma}$ at receivers $\text{RX}_k$, $k\in\left[j+\tau:K_R+j-1\right)]$. Towards this objective, precoders $\mathbf{V}_{j,\left[i:i+\tau-1\right]}^{\text{IA}}$, $i\in\left[K_T\right]$, are chosen to satisfy the following alignment conditions.
\begin{equation}
\begin{aligned}
\mathbf{M}_{\lbrace k\rbrace\cup\left[j+1:j+\tau-1\right]}^{\left[i:i+\tau-1\right]}\mathbf{V}_{j,\left[i:i+\tau-1\right]}^{\text{IA}} \prec \mathbf{V}_{j}^{\text{IA}},\qquad\forall i\in\left[K_T\right],\ \forall k\in\left[j+\tau:K_R+j-1\right],
\end{aligned}
\end{equation}    
where $\mathbf{P}\prec \mathbf{Q}$ means that the column space of matrix $\mathbf{P}$ is a subspace of the column space of the matrix $\mathbf{Q}$. First, we choose
\begin{equation}
\mathbf{V}_{j,\left[i:i+\tau-1\right]}^{\text{IA}}=\mathbf{V}_{j1}^{\text{IA}}.
\end{equation}
Therefore, the problem is reduced to finding precoders $\mathbf{V}_{j1}^{\text{IA}}$, $\mathbf{V}_{j}^{\text{IA}}$ for all $j\in\left[K_R\right]$ satisfying 
\begin{equation} ~\label{eqn5}
\begin{aligned}
\mathbf{M}_{\lbrace k\rbrace\cup\left[j+1:j+\tau-1\right]}^{\left[i:i+\tau-1\right]}\mathbf{V}_{j1}^{\text{IA}} \prec \mathbf{V}_{j}^{\text{IA}},\quad\forall i\in\left[K_T\right],\ \forall k\in\left[j+\tau:K_R+j-1\right].
\end{aligned}
\end{equation} 
In~\eqref{eqn5}, there are $\Gamma=K_T\left(K_R-\tau\right)$ conditions that $\mathbf{V}_{j1}^{\text{IA}}$ and $\mathbf{V}_{j}^{\text{IA}}$ are required to satisfy. First, we generate random vectors $\mathbf{a}_{j}=\left[a_{j}\left(1\right),\ldots,a_{j}\left(\mu_n\right)\right]^T$, $j\in\left[K_R\right]$, $i\in\left[K_T\right]$, where its elements are i.i.d. drawn from a continuous distribution. Then, we choose  
\begin{equation}~\label{eqn20}
\begin{aligned}
&\mathbf{V}_{j}^{\text{IA}}&&=\left\{\prod_{i=1}^{K_T}\prod_{k=j+\tau}^{K_R+j-1}\left(\mathbf{M}_{\lbrace k\rbrace \cup \left[j+1:j+\tau-1\right]}^{\left[i:i+\tau-1\right]}\right)^{\alpha_{j}^{\left[k,i\right]}}\mathbf{a}_{j}:0\leq\alpha_{j}^{\left[k,i\right]}\leq n\right\},\\
&\mathbf{V}_{j1}^{\text{IA}}&&=\left\{\prod_{i=1}^{K_T}\prod_{k=j+\tau}^{K_R+j-1}\left(\mathbf{M}_{\lbrace k\rbrace \cup \left[j+1:j+\tau-1\right]}^{\left[i:i+\tau-1\right]}\right)^{\alpha_{j}^{\left[k,i\right]}}\mathbf{a}_{j}:0\leq\alpha_{j}^{\left[k,i\right]}\leq n-1\right\}.
\end{aligned}
\end{equation}     
Thus, we can verify that the conditions in~\eqref{eqn5} are satisfied and the precoder matrices $\mathbf{V}_{j}^{\text{IA}}$ and $\mathbf{V}_{j1}^{\text{IA}}$, $j\in\left[K_R\right]$, have full column rank of $\left(n+1\right)^{\Gamma}$ and $n^{\Gamma}$, respectively. The proof is presented in Appendix~\ref{App4}. Now consider receiver $\text{RX}_k$. The desired streams of receiver $\text{RX}_k$ have arrived with directions of $K_Tn^{\Gamma}$ column vectors
\begin{equation}
\mathbf{D}_k=\left[\mathbf{M}_{\left[k:k+\tau-1\right]}^{\left[1:\tau\right]}\mathbf{V}_{k1}^{\text{IA}}\quad \mathbf{M}_{\left[k:k+\tau-1\right]}^{\left[2:\tau+1\right]}\mathbf{V}_{k1}^{\text{IA}}\quad\cdots\quad \mathbf{M}_{\left[k:k+\tau-1\right]}^{\left[K_T:K_T+\tau-1\right]}\mathbf{V}_{k1}^{\text{IA}}\right],
\end{equation}
while the interference signals have arrived after alignment with directions of $\left(K_R-\tau\right)\left(n+1\right)^{\Gamma}$ column vectors
\begin{equation}
\begin{aligned}
\mathbf{I}_k&=\left[\mathbf{V}_{j}^{\text{IA}}\right],\quad j\in\left[k+1:K_R+k-\tau\right] \\
&=\left[\mathbf{V}_{k+1}^{\text{IA}}\quad\cdots\quad \mathbf{V}_{K_R+k-\tau}^{\text{IA}}\right]
\end{aligned}
\end{equation}
To ensure that the receiver $\text{RX}_k$ can decode the desired streams, we should maintain that the directions of all desired streams are linearly independent of each other and independent of all directions of the interference streams. This can be ensured if the following matrix 
\begin{equation}~\label{eqn22}
\mathbf{R}_k=\left[\mathbf{D}_k\quad \mathbf{I}_k\right]
\end{equation}
has full rank of $\mu_n$ almost surely for all channel realizations. We left the proof of full rank to Appendix~\ref{App4}. As a result, each receiver $\text{RX}_k$ can decode $K_Tn^{\Gamma}$ desired streams over $\mu_n$-symbol extensions, and hence, a total of \small$\frac{K_RK_Tn^{\Gamma}}{K_Tn^{\Gamma}+\left(K_R-\tau\right)\left(n+1\right)^{\Gamma}}$ \normalsize DoF is achievable in each transmission block.

\section{Conclusion}~\label{Sec4}
In this paper, we have studied the degrees of freedom of interference networks with caches at the transmitters. Our main result is the characterization of the sum degrees of freedom of cache-aided interference channels within a multiplicative factor of $2$ from the information-theoretic bound, independent of all system parameters. To achieve this result, we have proposed a novel delivery scheme for arbitrary number of transmitters and receivers that outperforms state-of-the-art schemes in literature, for almost all transmitter cache sizes. The achievability proof depends mainly on scheduling the messages needed to be delivered to receivers into groups, where the groups are delivered in subsequent transmission blocks. Then, we construct two-layer precoding matrices based on ZF and IA techniques to deliver each group of messages in each transmission block. Moreover, we have derived an upper bound on the sum-DoF of delivery schemes that uses the uncoded caching scheme. We have shown that the derived upper bound is tighter than the upper bound that exists in the literature under small cache sizes. Even with these improvements, there still exists a multiplicative gap of at most a factor of $2$ between the achievable scheme and upper bounds. It is an interesting point to find more sophisticated bounds to narrow this gap.

\appendices
\section{Proof of Multiplicative Factor $2$}~\label{App2}
In this Appendix, we present the converse proof of Theorem~\ref{Th1}. We first derive a simple upper bound on the sum-DoF for general cache-aided networks without any restrictions on the caching and delivery schemes. Then, we show that the multiplicative gap between the achievable lower bound in Theorem~\ref{Th1} and the derived upper bound is less than $2$ independent of all system parameters. Assume that there exists full cooperation among $K_T$ transmitters to construct a single transmitter of $K_T$ antennas that has access to cache memories of all $K_T$ transmitters. Moreover, assume that $K_R$ receivers are allowed to fully cooperate with each other to construct a single receiver with $K_R$ antennas. Note that this cooperation can not reduce the sum-DoF, and hence, this assumption does not violate the upper bound. The constructed multi-antenna transmitter has the $K_R$ requested files of the single, multi-antenna receiver. Thus, the system becomes a point-to-point MIMO channel~\cite{zheng2003diversity}. Hence the total sum-DoF is bounded by
\begin{equation}
\mathsf{DoF}^{*}\leq \min\left\{K_T,K_R\right\}.
\end{equation}   
Let $\tau=K_TM/N$. When $\tau\geq \min\left\{K_T,K_R\right\}$, it is easily to see that both the upper and lower bounds are the same, i.e., the achievable scheme is optimal in this case.

When $1\leq \tau\leq\min\left\{K_T,K_R\right\}$, the proposed coding scheme achieves $\mathsf{DoF}=\frac{K_TK_R}{K_T+K_R-\tau}$. Hence, the multiplicative gap between the upper bound and the achievable DoF is given by
\begin{equation}
\begin{aligned}
\frac{\mathsf{DoF}^{*}}{\mathsf{DoF}}&= \min\left\{K_T,K_R\right\} / \frac{K_TK_R}{K_T+K_R-\tau}\leq 2.
\end{aligned}
\end{equation}
Therefore, the achievable sum-DoF is within a factor of $2$ from the derived upper bound.  

\vspace{-3mm}
\section{Converse of Cooperative $X$-network}~\label{App1}
Here, we present an upper bound on the sum-DoF of the cooperative $X$-network defined in Subsection~\ref{Subsec3}. The upper bound is mainly based on genie-aided and cut-set arguments. In such cooperative $X$-networks, we have $K_R{K_T\choose \tau}$ messages defined by $$\mathcal{W}\triangleq\left\{W_{j,\mathcal{S}}:j\in\left[K_R\right],\mathcal{S}\subseteq\left[K_T\right],|\mathcal{S}|=\tau\right\},$$  
such that each subset $\mathcal{S}\subseteq \left[K_T\right]$ of $\tau$ transmitters has a dedicated message to every receiver $\text{RX}_j$, $j\in\left[K_R\right]$. Let $R_{j,\mathcal{S}}\left(P\right)=\frac{|W_{j,\mathcal{S}}|}{T}$ denote the rate of the codeword encoding the message $W_{j,\mathcal{S}}$, where the DoF associated with message $W_{j,\mathcal{S}}$ is given by $d_{j,\mathcal{S}}=\lim_{P\to\infty}\frac{R^{*}_{j,\mathcal{S}}\left(P\right)}{\log\left(P\right)}.$   
Thus, the sum-DoF of the cooperative $X$-network is given by $\mathsf{DoF}_{\sum}=\sum_{\substack{\mathcal{S}\subseteq\left[K_T\right]\\|\mathcal{S}|=\tau}}\sum_{j\in\left[K_R\right]} d_{j,\mathcal{S}}$. Let $\mathbf{Y}_{\mathcal{K}}$ be a $|\mathcal{K}| T\times 1$ concatenated vector of the received signals of receivers $\text{RX}_j$, $j\in\mathcal{K}$, and $\mathbf{X}_{\mathcal{K}}$ be a $|\mathcal{K}| T\times 1$ concatenated vector of the transmitted signals of transmitters in set $\mathcal{K}$. Moreover, we consider $\mathbf{H}$ is a $K_R T\times K_T T$ channel matrix between transmitters and receivers over $T$ slots, and $\mathbf{H}_{\mathcal{K}_r}^{\mathcal{K}_t}$ is a $|\mathcal{K}_r|T\times |\mathcal{K}_t| T$ submatrix of $\mathbf{H}$.

\begin{lemma} The degrees of freedom of cooperative $X$-networks satisfies 
\begin{equation}~\label{eqn11}
\sum_{\substack{\mathcal{S}\subseteq\left[K_T\right]\\|\mathcal{S}|=\tau}}\sum_{j\in\mathcal{S}_r}d_{j,\mathcal{S}}+\sum_{\substack{\mathcal{S}\subseteq\mathcal{S}_t\\|\mathcal{S}|=\tau}}\sum_{j\in\mathcal{\overline{S}}_r}d_{j,\mathcal{S}}\leq \sigma
\end{equation} 
for a subset $\mathcal{S}_{t}$ of transmitters such that $|\mathcal{S}_{t}|=\sigma$ and a subset $\mathcal{S}_{r}$ of receivers such that $|\mathcal{S}_{r}|=\sigma$, where $\tau\leq\sigma\leq\min\left\{K_T,K_R\right\}$. Moreover $\mathcal{\overline{S}}_t=\left[K_T\right]\setminus\mathcal{S}_t$, and $\mathcal{\overline{S}}_r=\left[K_R\right]\setminus\mathcal{S}_r$. 
\end{lemma}
\begin{proof}

We first define the following three disjoint subsets of messages: 
\begin{equation}
\begin{aligned}
&\mathcal{W}_{\mathcal{S}_t}&&=\left\{W_{j,\mathcal{S}}: j\in\left[K_R\right],\mathcal{S}\subseteq\mathcal{S}_t, |\mathcal{S}|=\tau\right\},\\
&\mathcal{W}_{\mathcal{S}_r}&&=\left\{W_{j,\mathcal{S}}: j\in\mathcal{S}_r,\mathcal{S}\subseteq\left[K_T\right],\mathcal{S}\not\subseteq\mathcal{S}_t, |\mathcal{S}|=\tau\right\},\\
&\mathcal{\overline{\mathcal{W}}}&&=\left\{W_{j,\mathcal{S}}: j\in\mathcal{\overline{S}}_r\ \text{and}\ \mathcal{S}\not\subseteq\mathcal{S}_t, |\mathcal{S}|=\tau\right\}.
\end{aligned}
\end{equation}
Now assume that a genie provides receivers in set $\mathcal{S}_r$ with messages $\mathcal{\overline{W}}$, and receivers in set $\mathcal{\overline{S}}_r$ with messages $\mathcal{\overline{W}}\bigcup\mathcal{W}_{\mathcal{S}_r}$. Since the messages $\mathcal{\overline{W}}$ are given to all receivers, we are only concerned of the degrees of freedom of the remaining messages $\mathcal{W}_{\mathcal{S}_t}\bigcup\mathcal{W}_{\mathcal{S}_r}$. Assume that there is full cooperation between receivers $\text{RX}_j$, $j\in\mathcal{S}_r$. The main idea of the proof is to show that the set $\mathcal{S}_r$ of $\sigma$ receivers can decode the messages $\mathcal{W}_{\mathcal{S}_r}\bigcup \mathcal{W}_{\mathcal{S}_t}$ from its received signals using messages $\mathcal{\overline{W}}$ as a side information. We represent the received signals of $\mathcal{S}_r$ and $\mathcal{\overline{S}}_r$ receivers as 
\begin{equation}
\begin{aligned}
&\mathbf{Y}_{\mathcal{S}_r}&=\mathbf{H}_{\mathcal{S}_r}^{\mathcal{S}_t}\mathbf{X}_{\mathcal{S}_t}+\mathbf{H}_{\mathcal{S}_r}^{\mathcal{\overline{S}}_t}\mathbf{X}_{\mathcal{\overline{S}}_t}+\mathbf{Z}_{\mathcal{S}_r},\\
&\mathbf{Y}_{\mathcal{\overline{S}}_r}&=\mathbf{H}_{\mathcal{\overline{S}}_r}^{\mathcal{S}_t}\mathbf{X}_{\mathcal{S}_t}+\mathbf{H}_{\mathcal{\overline{S}}_r}^{\mathcal{\overline{S}}_t}\mathbf{X}_{\mathcal{\overline{S}}_t}+\mathbf{Z}_{\mathcal{\overline{S}}_r}.\\
\end{aligned}
\end{equation}  
Consider receivers in set $\mathcal{S}_r$ that can decode messages $\mathcal{W}_{\mathcal{S}_r}$. Using genie-aided messages $\mathcal{\overline{W}}$ and decoded messages $\mathcal{W}_{\mathcal{S}_r}$, receivers $j\in\mathcal{S}_r$ can compute the transmitted signals $\mathbf{X}_i$, $i\in\mathcal{\overline{S}}_t$ and subtract it from the received signal. Similarly, receivers in set $\mathcal{\overline{S}}_r$ have messages $\mathcal{\overline{W}}\bigcup \mathcal{W}_{\mathcal{S}_r}$ as a side information, and hence, they can compute the transmitted signals $\mathbf{X}_i$, $i\in\mathcal{\overline{S}}_t$ and subtract it from the received signal. Therefore, we can rewrite the signals of $\mathcal{S}_r$ and $\mathcal{\overline{S}}_r$ receivers as
\begin{equation}
\begin{aligned}
&\mathbf{\tilde{Y}}_{\mathcal{S}_r}&=\mathbf{H}_{\mathcal{S}_r}^{\mathcal{S}_t}\mathbf{X}_{\mathcal{S}_t}+\mathbf{Z}_{\mathcal{S}_r},\\
&\mathbf{\tilde{Y}}_{\mathcal{\overline{S}}_r}&=\mathbf{H}_{\mathcal{\overline{S}}_r}^{\mathcal{S}_t}\mathbf{X}_{\mathcal{S}_t}+\mathbf{Z}_{\mathcal{\overline{S}}_r},
\end{aligned}
\end{equation}
where receivers $j\in\mathcal{\overline{S}}_r$ are able to decode their intended messages $\lbrace W_{j,\mathcal{S}}:j\in\mathcal{\overline{S}}_r, \mathcal{S}\subseteq\mathcal{S}_t\rbrace$ from the received signal vector $\mathbf{\tilde{Y}}_{\mathcal{\overline{S}}_r}$, and receivers $j\in\mathcal{S}_r$ are able to decode their intended messages $\lbrace W_{j,\mathcal{S}}:j\in\mathcal{S}_r, \mathcal{S}\subseteq\mathcal{S}_t\rbrace$ from the received signal vector $\mathbf{\tilde{Y}}_{\mathcal{S}_r}$. Notice that the $\sigma T \times \sigma T$ submatrix channel $\mathbf{H}_{\mathcal{S}_r}^{\mathcal{S}_t}$ is invertable almost surely. Thus, by reducing noise at receivers $\mathcal{S}_r$ and multiplying the constructed signal $\mathbf{\tilde{Y}}_{\mathcal{S}_r}$ at receivers $\mathcal{S}_r$ by $\mathbf{H}_{\mathcal{\overline{S}}_r}^{\mathcal{S}_t}\left(\mathbf{H}_{\mathcal{S}_r}^{\mathcal{S}_t}\right)^{-1}$, we have 
\begin{equation}
\mathbf{\tilde{Y}}'_{\mathcal{S}_r}=\mathbf{H}_{\mathcal{\overline{S}}_r}^{\mathcal{S}_t}\mathbf{X}_{\mathcal{S}_t}+\mathbf{\tilde{Z}}_{\mathcal{\overline{S}}_r},
\end{equation}  
which is  a degraded version of $\mathbf{\tilde{Y}}_{\mathcal{\overline{S}}_r}$, where $\mathbf{\tilde{Z}}_{\mathcal{\overline{S}}_r}$ represents the reduced noise at receivers $\mathcal{S}_r$. Thus, receivers $\mathcal{S}_r$ can decode all messages $\mathcal{W}_{\mathcal{S}_t}$. Thus, by using Fano's inequality, we have
\begin{equation}
H\left(\mathcal{W}_{\mathcal{S}_t}|\mathbf{Y}_{\mathcal{S}_r},\mathcal{\overline{W}}\right)\leq H\left(\mathcal{W}_{\mathcal{S}_t}|\mathbf{Y}_{\mathcal{\overline{S}}_r},\mathcal{\overline{W}},\mathcal{W}_{\mathcal{S}_r}\right)\leq |\mathcal{W}_{\mathcal{S}_t}|T\epsilon.
\end{equation}
Notice that the applied assumptions (genie-aided information, cooperation between subset of receivers, reducing noise) cannot hurt the coding scheme. Thus, we have
\begin{IEEEeqnarray}{lll}~\label{Lower1}
H\left(\mathcal{W}_{\mathcal{S}_t},\mathcal{W}_{\mathcal{S}_r}\right)&=&\sum_{\substack{\mathcal{S}\subseteq\mathcal{S}_t\\|\mathcal{S}|=\tau}}\sum_{j\in\mathcal{\overline{S}}_r} R_{j,\mathcal{S}}^{*}T+\sum_{\substack{\mathcal{S}\subseteq\left[K_t\right]\\|\mathcal{S}|=\tau}}\sum_{j\in\mathcal{S}_r} R_{j,\mathcal{S}}^{*}T\nonumber\\
&\stackrel{\left(a\right)}{=}&H\left(\mathcal{W}_{\mathcal{S}_t},\mathcal{W}_{\mathcal{S}_r}|\mathcal{\overline{W}}\right)\nonumber\\
&\stackrel{\left(b\right)}{=}& I\left(\mathcal{W}_{\mathcal{S}_t},\mathcal{W}_{\mathcal{S}_r};\mathbf{Y}_{\mathcal{S}_r}|\mathcal{\overline{W}}\right)+H\left(\mathcal{W}_{\mathcal{S}_t},\mathcal{W}_{\mathcal{S}_r}|\mathbf{Y}_{\mathcal{S}_r},\mathcal{\overline{W}}\right)\\
&\stackrel{\left(c\right)}{\leq}& I\left(\mathbf{X}_{\left[K_T\right]};\mathbf{Y}_{\mathcal{S}_r}\right)+H\left(\mathcal{W}_{\mathcal{S}_t},\mathcal{W}_{\mathcal{S}_r}|\mathbf{Y}_{\mathcal{S}_r},\mathcal{\overline{W}}\right)\nonumber\\
&\stackrel{\left(d\right)}{\leq}& T\sigma\log\left(P\right)+H\left(\mathcal{W}_{\mathcal{S}_r}|\mathbf{Y}_{\mathcal{S}_r},\mathcal{\overline{W}}\right)+H\left(\mathcal{W}_{\mathcal{S}_t}|\mathbf{Y}_{\mathcal{S}_r},\mathcal{\overline{W}},\mathcal{W}_{\mathcal{S}_r}\right)\nonumber\\
&\stackrel{\left(e\right)}{\leq}& T\sigma\log\left(P\right)+|\mathcal{S}_r|T\epsilon+|\mathcal{S}_t|T\epsilon\nonumber
\end{IEEEeqnarray} 
where $\left(a\right)$ follows from the fact that the messages are independent. Step $\left(b\right)$ follows from the chain rule. Step $\left(c\right)$ follows from data processing inequality, where the signal $\mathbf{X}_{\left[K_T\right]}$ is a function of messages $\mathcal{W}_{\mathcal{S}_t}\bigcup\mathcal{W}_{\mathcal{S}_r}$. Step $\left(d\right)$ follows from the bound of the degrees of freedom of multiple access channel (MAC) with $K_T$ single-antenna transmitters and a receiver with $|\mathcal{S}_r|$ antennas. Finally, step $\left(e\right)$ follows from Fano's inequality. By diving on $T\log\left(P\right)$ and taking $P\to\infty$ and $\epsilon\to 0$, we get~\eqref{eqn11}.
\end{proof}

Repeating the inequality~\eqref{eqn11} for every subset $\mathcal{S}_r\subseteq\left[K_R\right]$ and every subset $\mathcal{S}_t\subseteq\left[K_T\right]$, or for simplicity consider a symmetric case where $\lbrace d_{j,\mathcal{S}}=d\rbrace$, we have
\begin{IEEEeqnarray}{lll} 
\mathsf{DoF}_{\sum}&\leq &\frac{\tau{K_T \choose \tau} K_R}{\tau {K_T \choose \tau}+\left(K_R-\sigma\right){\sigma-1 \choose \tau-1}}~\label{eqn12}
\end{IEEEeqnarray}  
To get the best tighter bound we optimize the bound in~\eqref{eqn12} on $\sigma$ such that $\tau\leq\sigma\leq\min\left\{K_T,K_R\right\}$. This completes the proof.
\vspace{-5mm}
\section{Proof of Lemma~\ref{Lem1}}~\label{App3}
For given receivers demand $\mathbf{d}$, we have $K_R{K_T\choose \tau}$ subfiles that need to be delivered to the receives. In~\eqref{eqn13}, we split each subfile into $\tau!\,\left(K_T-\tau\right)!$ smaller subfiles. Thus, a total of $K_R\,K_T!$ smaller subfiles needed to be delivered to all receivers. In Lemma~\ref{Lem1}, $K_R\,K_T! $ smaller subfiles are partitioned into groups such that each group of smaller subfiles is defined by a unique permutation $\pi\in\Pi_{\left[K_T\right]}^{\text{circ}}$ of total $\left(K_T-1\right)!$ groups. Furthermore inside each group, we have $ K_R K_T$ smaller subfiles. Accordingly, the total number of smaller subfiles in all groups is given by
\begin{equation}
 \left(K_T-1\right)!\, K_RK_T=K_R\,K_T! 
\end{equation}  
which is equal to the total number of smaller subfiles required to be delived, i.e., this partitioning strategy contains all smaller subfiles needed to be delivered. However, we should ensure that this groups contains disjoint smaller subfiles, and each smaller subfile is an actual smaller subfile needed to be delivered to maintain the correctness of the partitioning strategy.

From splitting strategy in~\eqref{eqn13}, we can see that each smaller subfile $W_{d_j,\pi\left[1:\tau\right],\pi\left[\tau+1:K_T\right]}^{\pi\left(1\right)}$, for $\pi\in\Pi_{\left[K_T\right]}$, $j\in\left[K_R\right]$, is required to be delivered to the receivers. Moreover, for a given circular permutation $\pi\in\Pi_{\left[K_T\right]}^{\text{circ}}$, if we circularly shift $\pi$ with $1\leq i< K_T$ times, it generates another permutation $\tilde{\pi}$, where $\tilde{\pi}\in\Pi_{\left[K_T\right]}$ but, $\tilde{\pi}\notin\Pi_{\left[K_T\right]}^{\text{circ}}$. Therefore, for a given circular permutation $\pi\in\Pi_{\left[K_T\right]}^{\text{circ}}$, smaller subfile $W_{d_j,\pi\left[i:i+\tau-1\right],\pi\left[i+\tau:K_T+i-1\right]}^{\pi\left(i\right)}$, for $i\in\left[K_T\right]$, $j\in\left[K_R\right]$, is an actual smaller subfiles required to be delivered to the receivers. Now, we prove that the groups are disjoint. To prove this, consider an arbitrary subfile $W_{d_j,\pi\left[i:i+\tau-1\right],\pi\left[i+\tau:K_T+i-1\right]}^{\pi\left(i\right)}$. We should ensure that this subfile does not appear at another group, where we prove this by contradiction. Assume that there exists a smaller subfile satisfying that  
\begin{equation}~\label{eqn19}
W_{d_j,\pi\left[i:i+\tau-1\right],\pi\left[i+\tau:K_T+i-1\right]}^{\pi\left(i\right)}\equiv W_{d_j,\tilde{\pi}\left[\tilde{i}:\tilde{i}+\tau-1\right],\tilde{\pi}\left[\tilde{i}+\tau:K_T+\tilde{i}-1\right]}^{\tilde{\pi}\left(\tilde{i}\right)}
\end{equation}
for $\tilde{\pi},\pi\in\Pi_{\left[K_T\right]}^{\text{circ}}$, $\tilde{\pi}\neq \mathcal{\pi}$, $1\leq i,\tilde{i}\leq K_T$. The condition in~\eqref{eqn19} is satisfied if and only if $\pi\left[i:K_T+i-1\right]=\tilde{\pi}\left[\tilde{i}:K_T+\tilde{i}-1\right]$. In other words, we can construct a circular permutation from another by circular shift process; however, this condition is not achievable. Therefore, each smaller subfile appears only at a single group, i.e., the groups are disjoint. Hence, the proof is completed. 
\vspace{-4mm}

\section{Full Rank of Matrices $\mathbf{V}_{j}^{\text{IA}}$ and $\mathbf{R}_k$}~\label{App4}
\vspace{-7mm}

In this Appendix, we prove that the matrices $\mathbf{V}_{k}^{\text{IA}}$, $\mathbf{V}_{k1}^{\text{IA}}$ have full column rank of $\left(n+1\right)^{\Gamma}$, $n^{\Gamma}$, respectively. Then, we prove that the matrix $\mathbf{R}_k$ has full rank of $\mu_n$. Before that, we present an important Lemma that is essential in the proofs.

\begin{lemma}~\label{Lem2} Let $\mathbf{H}$ be a $K_R\times K_T$ channel matrix whose entries are i.i.d. generated from a continuous distribution. For given $j\in\left[K_R\right]$, $\pi\in\Pi_{\left[K_T\right]}^{\text{circ}}$, and  $1\leq\tau\leq\min\left\{K_T,K_R\right\}$, any subset of monomials of $M_{\lbrace k\rbrace\cup\left[j+1:j+\tau-1\right]}^{\pi\left[i:i+\tau-1\right]}$ for $i\in\left[K_T\right]$, $k\in\left[j+\tau:K_R+j\right]$, are linearly independent.
\end{lemma}
\begin{proof}
Note that monomials of \small$\lbrace M_{\lbrace k\rbrace\cup\left[j+1:j+\tau-1\right]}^{\pi\left[i:i+\tau-1\right]}\rbrace_{i=1}^{K_T}$ \normalsize are linearly independent of monomials of \small$\lbrace M_{\lbrace l\rbrace\cup\left[j+1:j+\tau-1\right]}^{\pi\left[i:i+\tau-1\right]}\rbrace_{i=1}^{K_T}$ \normalsize for $l\in\left[j+\tau:K_R+j\right]$, $l\neq k$, since the channel coefficients $\lbrace h_{k\pi\left(i\right)}\rbrace_{i=1}^{K_T}$ only contribute in \small$\lbrace M_{\lbrace k\rbrace\cup\left[j+1:j+\tau-1\right]}^{\pi\left[i:i+\tau-1\right]}\rbrace_{i=1}^{K_T}$ \normalsize and the channels are i.i.d. variables. Thus, it is only required to prove the linearly independence of monomials of \small$\lbrace M_{\lbrace k\rbrace\cup\left[j+1:j+\tau-1\right]}^{\pi\left[i:i+\tau-1\right]}\rbrace_{i=1}^{K_T}$. \normalsize Observe that each minor \small$ M_{\lbrace k\rbrace\cup\left[j+1:j+\tau-1\right]}^{\pi\left[i:i+\tau-1\right]}$ \normalsize is a polynomial of degree $\tau$. Therefore, monomials of \small$\lbrace M_{\lbrace k\rbrace\cup\left[j+1:j+\tau-1\right]}^{\pi\left[i:i+\tau-1\right]}\rbrace_{i=1}^{K_T}$ \normalsize with different degrees are linearly independent due to the different exponents of the channel coefficients. This implies that we only need to prove the linearly independence of monomials of \small$ M_{\lbrace k\rbrace\cup\left[j+1:j+\tau-1\right]}^{\pi\left[i:i+\tau-1\right]}$ \normalsize with the same degree. Consider a linear combination of a set of $N$ monomials of degree $L$. 
\small \begin{equation}~\label{eqn24}
\sum_{n=1}^{N} c_n \prod_{i=1}^{K_T}\left(M_{\lbrace k\rbrace\cup\left[j+1:j+\tau-1\right]}^{\pi\left[i:i+\tau-1\right]}\right)^{\alpha_i\left(n\right)}=0
\end{equation}
\normalsize where $\sum_{i=1}^{K_T}\alpha_i\left(n\right)=L$, and $\mathbf{c}=\left[c_1,\ldots,c_N\right]\in\mathbb{C}^{N}$. If these monomials are not linearly independent, then, there exists $\mathbf{c}\neq \mathbf{0}$ satisfying~\eqref{eqn24}. Note that the linear combination in~\eqref{eqn24} can be factored into multiplied terms, where each term is a linear combination of \small $\lbrace M_{\lbrace k\rbrace\cup\left[j+1:j+\tau-1\right]}^{\pi\left[i:i+\tau-1\right]}\rbrace_{i=1}^{K_T}$. \normalsize As a result monomials of \small $\lbrace M_{\lbrace k\rbrace\cup\left[j+1:j+\tau-1\right]}^{\pi\left[i:i+\tau-1\right]}\rbrace_{i=1}^{K_T}$ \normalsize of the same degree are linearly independent if and only if the minors \small $\lbrace M_{\lbrace k\rbrace\cup\left[j+1:j+\tau-1\right]}^{\pi\left[i:i+\tau-1\right]}\rbrace_{i=1}^{K_T}$ \normalsize are linearly independent. Using cofactor expansion, we represent \small $M_{\lbrace k\rbrace\cup\left[j+1:j+\tau-1\right]}^{\pi\left[i:i+\tau-1\right]}=\sum_{l=i}^{i+\tau-1}h_{k\pi\left(l\right)}C_{\left[j+1:j+\tau-1\right]}^{\pi\left[i:i+\tau-1\right]\setminus \pi\left(l\right)}$, \normalsize where \small $C_{\left[j+1:j+\tau-1\right]}^{\pi\left[i:i+\tau-1\right]\setminus \pi\left(l\right)}$ \normalsize is the cofactor of \small$h_{k\pi\left(l\right)}$. \normalsize Thus, the linear combination of these minors is given by
\small \begin{equation}~\label{eqn25}
\sum_{i=1}^{K_T}c_i M_{\lbrace k\rbrace\cup\left[j+1:j+\tau-1\right]}^{\pi\left[i:i+\tau-1\right]}=\sum_{l=1}^{K_T}h_{k\pi\left(l\right)}\left(\sum_{i=l-\tau+1}^{l} c_i C_{\left[j+1:j+\tau-1\right]}^{\pi\left[i:i+\tau-1\right]\setminus \pi\left(l\right)}\right)=0
\end{equation}
\normalsize Since \small $\lbrace h_{k\pi\left(l\right)}\rbrace_{l=1}^{K_T}$ \normalsize are i.i.d. variables,~\eqref{eqn25} is satisfied if and only if \small $\sum_{i=l-\tau+1}^{l} c_i C_{\left[j+1:j+\tau-1\right]}^{\pi\left[i:i+\tau-1\right]\setminus \pi\left(l\right)}=0$ \normalsize for every $l\in\left[K_T\right]$. Note that if $\tau=2$, then \small $\lbrace C_{\left[j+1:j+\tau-1\right]}^{\pi\left[i:i+\tau-1\right]\setminus \pi\left(l\right)}\rbrace_{i=l-\tau+1}^{l}$ \normalsize are i.i.d. channel coefficients. Otherwise, we can use cofactor expansion to represent \small $\lbrace C_{\left[j+1:j+\tau-1\right]}^{\pi\left[i:i+\tau-1\right]\setminus \pi\left(l\right)}\rbrace_{i=l-\tau+1}^{l}$ \normalsize. Then, the same argument can be iteratively applied until we reach to a linear combination of some channels that are i.i.d. random variables. Therefore, we reach to the result that $\lbrace c_i\rbrace_{i=1}^{K_T}$ must be zero to satisfy~\eqref{eqn25}, i.e., \small $\lbrace M_{\lbrace k\rbrace\cup\left[j+1:j+\tau-1\right]}^{\pi\left[i:i+\tau-1\right]}\rbrace_{i=1}^{K_T}$ \normalsize are linearly independent. The details are omitted here due to the space limitations.
\end{proof}
   
Consider first the $\mu_n\times \left(n+1\right)^\Gamma$ matrix $\mathbf{V}_j^{\text{IA}}$ given in~\eqref{eqn20}. At any row $u\in\left[\mu_n\right]$, the entries of different columns are different monomials in \small$\left\{M_{\lbrace k\rbrace\cup\left[j+1:j+\tau-1\right]}^{\pi\left[i:i+\tau-1\right]}\left(u\right)\right\}$ \normalsize for \small $i\in\left[K_T\right]$, $k\in\left[j+\tau:K_R+j-1\right]$. \normalsize In other words, \small $\left[\alpha_{j}^{\left[k,i\right]}\left(l\right)\right]_{k\in\left[j+\tau:K_R+j-1\right]}^{i\in\left[K_T\right]}\neq\left[\alpha_{j}^{\left[k,i\right]}\left(l'\right)\right]_{k\in\left[j+\tau:K_R+j-1\right]}^{i\in\left[K_T\right]}$ \normalsize for \small $\left(l,l'\right)\in\left[\left(n+1\right)^\Gamma\right]$ \normalsize with $l\neq l'$. Furthermore, the exponents of monomials at the same column with different rows are equal. Thus, using~\cite[Lemma~$3$]{annapureddy2012degrees}, the matrix $\mathbf{V}_{j}^{\text{IA}}$ has full rank almost surely if and only if \small$\left\{M_{\lbrace k\rbrace\cup\left[j+1:j+\tau-1\right]}^{\pi\left[i:i+\tau-1\right]}\right\}$ \normalsize for  $i\in\left[K_T\right]$, $k\in\left[j+\tau:K_R+j-1\right]$, are linearly independent. Using Lemma~\ref{Lem2}, matrix $\mathbf{V}_{j}^{\text{IA}}$ has full rank of $\left(n+1\right)^\Gamma$. Similarly, we can prove that matrix $\mathbf{V}_{j1}^{\text{IA}}$ has full rank of $n^{\Gamma}$, where it is designed in the same manner as matrix $\mathbf{V}_{j}^{\text{IA}}$. Now, consider matrix $\mathbf{R}_{k}$ given in~\eqref{eqn22}. Using~\cite[Lemma~$3$]{annapureddy2012degrees}, matrix $\mathbf{R}_k$ is full rank if and only if there does not exist an annihilating polynomial satisfying

\small
\begin{equation}\label{eqn23}
\sum_{i=1}^{K_T}\sum_{l=1}^{n^\Gamma} c_{kil} M_{\left[k:k+\tau-1\right]}^{\pi\left[i:i+\tau-1\right]} V_{k1}^{\text{IA}}\left(l\right)+\sum_{j=k+1}^{K_R+k-\tau}\sum_{l=1}^{\left(n+1\right)^\Gamma} c_{jl} V_{j}^{\text{IA}}\left(l\right)=0,
\end{equation}
\normalsize where $\lbrace c_{kil}\rbrace$ and $\lbrace c_{jl}\rbrace$ are non-zero factors. $V_{k1}^{\text{IA}}\left(l\right)$ is the entry in the $l$-th column and an arbitrary row of matrix $\mathbf{V}_{k1}^{\text{IA}}$, and $V_{j}^{\text{IA}}\left(l\right)$ is the entry in the $l$-th column and an arbitrary row of matrix $\mathbf{V}_{j}^{\text{IA}}$. Note that the variable $a_{j}$ appears with exponent one in the monomials $\lbrace V_{j}^{\text{IA}}\left(l\right)\rbrace_{l=1}^{\left(n+1\right)^\Gamma}$ and does not appear in the other monomials. Therefore, these monomials are linearly independent of other monomials in~\eqref{eqn23}, since $\lbrace a_j\rbrace_{j=1}^{K_R}$ are i.i.d. random variables. As a result, we need only to prove that the monomials of the first term in~\eqref{eqn23} are linearly independent. Observe that the first term is monomials in \small $\lbrace M_{\lbrace j\rbrace\cup\left[k+1:k+\tau-1\right]}^{\left[i:i+\tau-1\right]}\rbrace$ \normalsize for $i\in\left[K_T\right]$, $j\in\left[k+\tau:K_R+k\right]$. Using Lemma~\ref{Lem2}, the monomials of these minors are linearly independent. As a result matrix $\mathbb{R}_k$ has full rank of $\mu_n$.

\end{document}